\documentclass[11pt]{article}

\usepackage{amsfonts,latexsym,amsthm,amssymb,amsmath,amscd,euscript,tikz,mathtools}
\usepackage{framed}
\usepackage[margin=1in]{geometry}
\usepackage{color}
\usepackage[colorlinks=true,citecolor=blue,linkcolor=blue]{hyperref}
\usepackage{enumitem}
\usepackage{siunitx}
\usepackage{textcomp}
\usepackage{physics}
\usepackage{mathrsfs}
\usepackage{dsfont}
\usepackage[ruled,vlined,linesnumbered]{algorithm2e}
\usepackage{titlesec}
\usepackage{tikz-cd}
\usepackage{thmtools}
\usepackage{thm-restate}
\usepackage{cleveref}
\usepackage{graphicx}
\usetikzlibrary{positioning,chains,fit,shapes,calc}

\allowdisplaybreaks[1]

\newtheorem{theorem}{Theorem}
\newtheorem{proposition}[theorem]{Proposition}
\newtheorem{lemma}[theorem]{Lemma}
\newtheorem{claim}[theorem]{Claim}

\theoremstyle{definition}
\newtheorem{definition}[theorem]{Definition}
\newtheorem{example}[theorem]{Example}
\newtheorem{remark}[theorem]{Remark}

\theoremstyle{remark}

\newcommand{\cA}{\mathcal{A}}\newcommand{\cB}{\mathcal{B}}
\newcommand{\cC}{\mathcal{C}}

\newcommand{\bE}{\mathbb{E}}\newcommand{\bF}{\mathbb{F}}

\newcommand{\bN}{\mathbb{N}}

\newcommand{\bR}{\mathbb{R}}

\newcommand{\bZ}{\mathbb{Z}}

\newcommand{\1}{\mathds{1}}

\newcommand{\poly}{\operatorname{poly}}

\newcommand{\Maj}{\operatorname{Majority}}
\newcommand{\llog}{\eta}

\newcommand{\nc}{\newcommand}

\nc{\on}{\operatorname}
\nc{\Spec}{\on{Spec}}
\nc{\Aut}{\textit{Aut}}
\nc{\id}{\textit{id}}
\nc{\chr}{\on{char}}
\nc{\im}{\on{im}}
\nc{\Hom}{\on{Hom}}
\nc{\lcm}{\on{lcm}}
\nc{\dual}[1]{\prescript{t}{}{#1}}
\nc{\transpose}[1]{{#1}^{\intercal}}
\nc{\Sym}{\on{Sym}}
\nc{\End}{\on{End}}
\nc{\stab}{\on{stab}}
\nc{\Li}{\on{Li}}
\nc{\spn}{\on{span}}
\nc{\sgn}{\on{sgn}}
\nc{\supp}{\on{supp}}
\nc{\Unif}{\on{Unif}}

\title{
  Decoding Quasi-Cyclic Quantum LDPC Codes\thanks{Research supported in part by a ONR grant N00014-24-1-2491, a UC Noyce initiative award, and a Simons Investigator award. L.~Golowich is supported by a National Science Foundation Graduate Research Fellowship under Grant No.~DGE 2146752.} 
}
\author{Louis Golowich \\
  UC Berkeley \\
  \href{mailto:lgolowich@berkeley.edu}{\texttt{lgolowich@berkeley.edu}}
  \and
  Venkatesan Guruswami \\
  UC Berkeley \\
  \href{mailto:venkatg@berkeley.edu}{\texttt{venkatg@berkeley.edu}}
}

\begin{document}

\pagenumbering{gobble}

\maketitle
\thispagestyle{empty}

\begin{abstract}
  Quantum low-density parity-check (qLDPC) codes are an important component in the quest for quantum fault tolerance. Dramatic recent progress on qLDPC codes has led to constructions which are asymptotically good, and which admit linear-time decoders to correct errors affecting a constant fraction of codeword qubits~\cite{panteleev_asymptotically_2022,leverrier_decoding_2023,dinur_good_2023,gu_efficient_2023}. These constructions, while theoretically explicit, rely on inner codes with strong properties only shown to exist by probabilistic arguments, resulting in lengths that are too large to be practically relevant. In practice, the surface/toric codes, which are the product of two repetition codes, are still often the qLDPC codes of choice.

  A construction preceding~\cite{panteleev_asymptotically_2022} based on the lifted product of an expander-based classical LDPC code with a repetition code~\cite{panteleev_quantum_2022} achieved a near-linear distance (of $\Omega(N/\log N)$ where $N$ is the number of codeword qubits), and avoids the need for such intractable inner codes. Our main result is an efficient decoding algorithm for these codes that corrects $\Theta(N/\log N)$ adversarial errors. En route, we give such an algorithm for the hypergraph product version these codes, which have weaker $\Theta(\sqrt{N})$ distance (but are simpler). Our decoding algorithms leverage the fact that the codes we consider are \textit{quasi-cyclic}, meaning that they respect a cyclic group symmetry.

  Since the repetition code is not based on expanders, previous approaches to decoding expander-based qLDPC codes, which typically worked by greedily flipping code bits to reduce some potential function, do not apply in our setting. Instead, we reduce our decoding problem (in a black-box manner) to that of decoding classical expander-based LDPC codes under noisy parity-check syndromes. For completeness, we also include a treatment of such classical noisy-syndrome decoding that is sufficient for our application to the quantum setting.
\end{abstract}

\newpage

\tableofcontents

\newpage

\pagenumbering{arabic}

\section{Introduction}
\label{sec:intro}
Quantum low-density parity-check (qLDPC) codes provide a particularly promising avenue for achieving low-overhead quantum error correction. Alongside code design, one of the most important challenges associated with qLDPC codes is the design of efficient decoding algorithms, which are needed to correct errors more quickly than they occur in quantum devices. In this paper, we present new efficient algorithms to decode qLDPC codes arising from a product of a classical LDPC code and a repetition code against a number of adversarial errors growing linearly in the code distance. As described below, the problem of decoding such quantum codes arise in various settings of interest, and even permits a purely classical interpretation. To the best of our knowledge, our results provide the first known polynomial-time decoding algorithms for these codes.

A qLDPC code is a quantum code for which errors on code states can be corrected by first performing some constant-weight quantum stabilizer (i.e.~parity-check) measurements, and then running a classical decoding procedure that processes these measurement outcomes to identify an appropriate correction that will revert the effect of the error. The stabilizer measurements are typically efficient\footnote{Assuming an architecture permitting arbitrary qubit connectivity.}, at least in theory, as they by definition each only involve a constant number of qubits. Therefore a key challenge in decoding qLDPC codes is to define a classical decoding algorithm, which receives as input the stabilizer measurement outcomes, called the \textit{syndrome}. This algorithm must output a Pauli correction that can be applied to revert the effect of the physical error on the code qubits.

We typically want to design efficient decoding algorithms for qLDPC codes with large \textit{distance} and \textit{dimension} relative to the \textit{length}, or number of physical qubits. The distance of a code measures the maximum number of corruptions from which the original code state can be recovered information-theoretically (but possibly inefficiently). In particular, codes of distance $D$ permit information-theoretic error-correction from adversarial errors acting on up to $(D-1)/2$ physical code qubits. Meanwhile, the dimension of a code equals the number of logical qubits that can be encoded in the code.

For over two decades, there were no known qLDPC codes of length $N$ and distance above $\tilde{O}(\sqrt{N})$, which up to poly-logarithmic factors is the distance achieved by Kitaev's Toric code \cite{kitaev_fault-tolerant_2003}. However, a recent breakthrough line of work achieved the first qLDPC codes of distance $N^{1/2+\Omega(1)}$ \cite{hastings_fiber_2021,breuckmann_balanced_2021,hastings_quantum_2023}, culminating first in qLDPC codes of nearly linear distance $D=\Theta(N/\log N)$ with dimension $K=\Theta(\log N)$ \cite{panteleev_quantum_2022}, and then finally of linear distance $D=\Theta(N)$ and dimension $K=\Theta(N)$ (i.e.~\textit{asymptotically good}) \cite{panteleev_asymptotically_2022}. The construction of \cite{panteleev_asymptotically_2022}, along with some spinoff constructions \cite{leverrier_quantum_2022-1,dinur_good_2023}, remain the only known linear-distance qLDPC codes. Meanwhile, the codes of \cite{panteleev_quantum_2022} remain the only other known construction of qLDPC codes with almost linear distance $D\geq N^{1-o(1)}$.

The codes of \cite{panteleev_asymptotically_2022,leverrier_quantum_2022-1,dinur_good_2023}, though asymptotically good in theory, are constructed using inner codes for which the only known constructions are randomized and prohibitively large in practice. In contrast, the nearly-linear-distance codes of \cite{panteleev_quantum_2022} are constructed as a product of an expander-based classical LDPC code with a repetition code, and thereby avoid the need for such intractable inner code properties. These codes of \cite{panteleev_quantum_2022} are said to be \textit{quasi-cyclic}, as they respect the action of a large cyclic group, in part due to the cyclic symmetry of the repetition code.

The asymptotically good codes of \cite{panteleev_asymptotically_2022,leverrier_quantum_2022-1,dinur_good_2023} have been shown to have linear-time decoders \cite{gu_efficient_2023,leverrier_decoding_2023,dinur_good_2023}, which can be parallelized to run in logarithmic time \cite{leverrier_decoding_2023}. However, to the best of our knowledge, there were no known polynomial-time decoding algorithms for the nearly-linear distance quasi-cyclic codes of \cite{panteleev_quantum_2022}. Such a decoder is our principal contribution:

\begin{theorem}[Informal statement of Theorem~\ref{thm:lpdec} with Proposition~\ref{prop:classtan}]
  \label{thm:lpdecinf}
  There exists an explicit family of qLDPC codes $\cC$ of length $N\rightarrow\infty$, dimension $K=\Theta(\log N)$, and distance $D=\Theta(N/\log N)$ obtained by instantiating the construction of \cite{panteleev_quantum_2022}, such that the following holds. Fix any constants $0<\epsilon\leq 1/2$ and $\epsilon'>0$. Then there is a $O(N^{2+O(\epsilon+\epsilon')})$-time randomized decoding algorithm for $\cC$ that successfully decodes against $\Theta(\epsilon D)$ adversarial errors with probability $\geq 1-2^{-N^{\epsilon'}}$.
\end{theorem}

In words, Theorem~\ref{thm:lpdecinf} provides a close-to-quadratic time randomized algorithm that decodes the codes of \cite{panteleev_quantum_2022} from a linear number of errors with respect to the code distance, with just an exponentially small probability of a decoding failure. At a very high level, our algorithm works by making appropriate black-box calls to a classical LDPC code decoder that can handle errors on both code bits and parity-check bits.

We emphasize that Theorem~\ref{thm:lpdecinf} provides a decoding algorithm that corrects against arbitrary adversarial errors, not simply randomized errors. The randomization and exponentially small failure probability in Theorem~\ref{thm:lpdecinf} instead arise from a randomized component of the decoding algorithm. As shown in Section~\ref{sec:lpdec}, we depress the failure probability to be exponentially small using standard amplification techniques; we can similarly detect decoding failures when they do occur.

\subsection{Overview of QLDPC Code Constructions}
\label{sec:introconstruct}
The nearly linear-distance quasi-cyclic qLDPC codes of \cite{panteleev_quantum_2022} in Theorem~\ref{thm:lpdecinf} are constructed as the \textit{lifted product} (see Section~\ref{sec:lpdef}) of a classical Tanner code (on an expander graph) with a repetition code. Lifted products provide a general means for constructing a quantum code from two classical codes that respect a group symmetry.

The asymptotically good qLDPC codes of \cite{panteleev_asymptotically_2022,leverrier_quantum_2022-1,dinur_good_2023} are also constructed using a lifted product, but of two classical Tanner codes on expander graphs (see also \cite{tillich_quantum_2014}). Meanwhile, the Toric code \cite{kitaev_fault-tolerant_2003} is constructed as a product of two repetition codes. We remark that the Toric code, along with the closely related surface code \cite{bravyi_quantum_1998}, has dimension $K=\Theta(1)$, distance $\Theta(\sqrt N)$, was essentially the state-of-the-art qLDPC code for many years, and remains integral to practical quantum error correction efforts.

To summarize, many qLDPC codes of interest in the literature are constructed either as a product of two expander-based classical LDPC codes, or as a product of two repetition codes. The quasi-cyclic codes of \cite{panteleev_quantum_2022} interpolate between these two constructions by taking a product of an expander-based classical LDPC code with a repetition code.

Our work initiates the study of efficiently decoding such products of an expander-based code with a repetition code. Our techniques are applicable beyond the decoder in Theorem~\ref{thm:lpdecinf} for the codes of \cite{panteleev_quantum_2022}. For instance, we show that a similar, but less involved, decoding algorithm applies to another class of quasi-cyclic codes, obtained as a related (but different) product, namely the hypergraph product, of an expander-based classical LDPC code with a repetition code (see Section~\ref{sec:hgpdec}). The reader is referred to Section~\ref{sec:introcomparison} for comparisons to prior decoders for related codes.

Such hypergraph products of expander-based codes with repetition codes arise naturally in various settings. For instance, similar such codes are used to teleport logical qubits between a product of two expander-based codes and a surface code; see for instance \cite{xu_constant-overhead_2023}.

Perhaps more fundamentally, the hypergraph product of some code $\cA$ with a repetition code $\cB$ can be viewed as a description of how the code $\cA$ handles errors occuring on different code components at different points in time. At a high level, here the code $\cA$ describes the available storage space, while the repetition code $\cB$ serves as the time axis. This viewpoint is often used in quantum error correction to design decoding algorithms, such as for the surface code, that are robust against syndrome errors (see e.g.~\cite{dennis_topological_2002}). Appendix~\ref{sec:spacetime} applies this viewpoint to give a fully classical interpretation of our decoding algorithm for quantum hypergraph product codes.

\subsection{Overview of Decoding Techniques}
\label{sec:introdecoverview}
Our decoding algorithm for the codes $\cC$ in Theorem~\ref{thm:lpdecinf} leverages the group symmetry that $\cC$ respects due to its quasi-cyclic nature. If we arrange the bits of a length-$\ell$ repetition code in a circle, then the code is preserved under the cyclic group action of $\bZ/\ell\bZ$, where $i\in\bZ/\ell\bZ$ simply shifts all bits by $i$ positions clockwise. As $\cC$ is obtained as a product of an expander-based code (also respecting a cyclic group action) with a length-$\ell$ repetition code, $\cC$ inherits this action of $\bZ/\ell\bZ$. Therefore the code space of $\cC$ is preserved under permutations of the qubits by actions of the group $\bZ/\ell\bZ$, so $\cC$ is indeed a quasi-cyclic quantum code.

At a high level, our decoder in Theorem~\ref{thm:lpdecinf} repeatedly adds together different cyclic shifts of the received syndrome, and then passes this sum through a classical Tanner code decoder to gain information about the original error. Our algorithm can use, in a black-box manner, any such classical decoder that corrects against errors on the code bits in the presence of errors on the parity-check syndrome bits. The main idea is that we isolate the effects of certain errors by summing together different cyclic shifts of the syndrome; we can then identify, and ultimately correct, these errors. \cite{panteleev_quantum_2022} used a related idea to prove the distance bound in these codes.

However, our actual algorithm is rather delicate, and requires an iterative aspect that is absent from the distance proof of \cite{panteleev_quantum_2022}. Specifically, our algorithm consists of $\Theta(\log\ell)$ iterations, where the $\tau$th iteration adds together $2^\tau$ cyclic shifts of the syndrome, and uses a classical Tanner code decoder to perform some analysis with a constant probability of failure. With probability $2^{-\Theta(\log\ell)}=1/\poly(\ell)$, all iterations succeed, and we decode successfully. We then amplify the success probability by repeating this procedure $\poly(\ell)$ times.

\subsection{Comparison to Prior Decoders}
\label{sec:introcomparison}
Our techniques described above differ from those used in decoders of \cite{gu_efficient_2023,leverrier_decoding_2023,dinur_good_2023} for asymptotically good qLDPC codes \cite{panteleev_asymptotically_2022,leverrier_quantum_2022-1,dinur_good_2023}. At a high level, given a syndrome, these decoders recover the error by greedily flipping code bits to reduce some potential function related to the syndrome weight. Such greedy ``flip''-style decoders are often applied for expander-based classical and quantum LDPC codes (e.g.~\cite{sipser_expander_1996,zemor_expander_2001,leverrier_quantum_2015,fawzi_efficient_2018}). However, beacuse the quasi-cyclic codes of \cite{panteleev_quantum_2022} that we consider in Theorem~\ref{thm:lpdecinf} are constructed with products involving repetition codes, which are not based on expanders, the resulting codes $\cC$ do not seem to be ``sufficiently expanding'' to permit such a greedy flip-style decoder.

All previously known decoders \cite{gu_efficient_2023,leverrier_decoding_2023,dinur_good_2023} for linear (or almost linear) distance qLDPC codes \cite{panteleev_asymptotically_2022,leverrier_quantum_2022-1,dinur_good_2023} require the code to be constructed from a constant-sized pair of ``inner codes'' satisfying a property called \textit{(two-way) product-expansion} (also known as \textit{robust testability}; see e.g.~\cite{kalachev_two-sided_2023,dinur_good_2023}). Such constant-sized objects have been shown to exist via probabilistic arguments. However, the constants are prohibitively large to the point where no specific instances of such two-way product-expanding codes have been constructed, to the best of our knowledge.

In contrast, the codes of \cite{panteleev_quantum_2022} do not need product-expanding inner codes, and hence may be more relevant for practical implementations. Furthermore, as mentioned in Section~\ref{sec:introdecoverview}, our decoding algorithm in Theorem~\ref{thm:lpdecinf} is based on the black-box application of classical LDPC code decoders in the quantum setting, which may be helpful for implementation purposes. In contrast, the flip-style decoders of \cite{gu_efficient_2023,leverrier_decoding_2023,dinur_good_2023} are more inherently quantum, as they do not naturally reduce to decoding a good classical code.

For hypergraph products of arbitrary classical codes, the ReShape decoder of \cite{quintavalle_reshape_2022} can also correct a linear number adversarial errors with respect to the code distance, and also is based on a black-box application of classical decoders. While we restrict attention to products of expander-based codes with repetition codes, we are able to decode lifted products of nearly linear distance, whereas \cite{quintavalle_reshape_2022} only consider hypergraph products (which have distance at most $O(\sqrt{N})$ for codes of length $N$). Furthermore, for hypergraph products, we present a decoding algorithm that runs in almost linear time with respect to the block length (see Section~\ref{sec:hgpdecinf}), whereas ReShape requires quadratic running time \cite{quintavalle_reshape_2022}.

One downside of the codes of \cite{panteleev_quantum_2022} that we consider is their logarithmic dimension $K=\Theta(\log N)$, which is much worse than the optimal $K=\Theta(N)$. We can boost the dimension at the cost of degrading the distance by simply breaking the message into equal-sized blocks, and encoding each block into a smaller instance of the codes of \cite{panteleev_quantum_2022}. Indeed, for any $0<\alpha<1$, we can obtain length-$N$ qLDPC codes of dimension $K=\Theta(N^\alpha\log N)$ and distance $D=\Theta(N^{1-\alpha}/\log N)$ by simply using $N^\alpha$ copies of the length-$N^{1-\alpha}$ code of \cite{panteleev_quantum_2022}. Our decoder in Theorem~\ref{thm:lpdecinf} by definition still applies to such codes. \cite{panteleev_quantum_2022} showed how to improve these parameters to $K=\Theta(N^\alpha\log N)$ and $D=\Theta(N^{1-\alpha/2}/\log N)$ using an additional product construction, though we have not analyzed the decoding problem for these codes. It is an interesting question whether the dimension of such constructions could be further increased, while preserving the distance and efficient decodability, and without requiring a stronger ingredient such as product-expansion.

\subsection{Open Problems}
Our work raises a number of open problems, such as those listed below:
\begin{itemize}
\itemsep=0ex
\item Can our decoding algorithms be adapted to correct a number of \textit{random} errors growing linearly in the block length?
\item Can our decoding algorithms be made to accomodate syndrome measurement errors? Note that our algorithms use as an ingredient classical codes robust to syndrome errors, which is distinct from considering errors on the syndrome of the quantum code.
\item Can our decoding algorithms be parallelized? We see no impediment to parallelization, but for simplicity in the presentation we did not pursue this direction.
\item Can our techniques be used to improve decoder performance in practical implementations? Our black-box use of classical LDPC decoders may be useful in this regard.
\item Can our techniques be extended to more general classes of codes, or to codes with improved paramters? See for instance the discussion in Section~\ref{sec:introcomparison}.
\end{itemize}

\subsection{Roadmap}
The remainder of this paper is organized as follows. Section~\ref{sec:prelim} provides necessary preliminary notions on classical and quantum codes. For completeness Section~\ref{sec:prelim} is fairly detailed, but readers familiar with quantum LDPC codes and the language of chain complexes may skip most of this background and head directly to Section~\ref{sec:tecoverview}, which provides a technical overview of our decoding algorithms. Section~\ref{sec:classtan} describes the classical LDPC codes that we use in a black-box manner for the quantum codes we consider. Section~\ref{sec:hgpdec} presents our decoder for the hypergraph product of an expander-based classical LDPC code with a repetition code. This hypergraph product decoder provides a helpful introduction to some of the basic techniques involved in the proof of our main result (Theorem~\ref{thm:lpdecinf}) on decoding lifted products of such classical codes. We prove Theorem~\ref{thm:lpdecinf} in Section~\ref{sec:lpdec}.

Appendix~\ref{sec:spacetime} provides a purely classical interpretation of our hypergraph product decoder. This perspective may be particularly helpful for those more familiar with classical error correction.

Appendix~\ref{sec:noisysyn} provides the formal construction of the expander-based classical LDPC codes we describe in Section~\ref{sec:classtan}, and proves the necessary properties. In particular, Appendix~\ref{sec:noisysyn} provides a self-contained presentation of efficient decoding of classical LDPC codes, along with their ``transpose'' codes, under errors on both the code bits and syndrome bits.

\section{Preliminaries}
\label{sec:prelim}
This section presents preliminary notions and relevant prior results pertaining to classical and quantum codes. For simplicity, in this paper we restrct attention to codes over the binary alphabet $\bF_2$.

\subsection{Basic Notation}
For an integer $n\geq 0$, we let $[n]=\{0,\dots,n-1\}$ denote the set of nonnegative integers less than $n$.

\subsection{Chain Complexes}
It will be most convenient for us to phrase both classical and quantum codes in the (standard) language of chain complexes.

\begin{definition}
  \label{def:chaincomplex}
  A \textbf{$t$-term chain complex} (over $\bF_2$) $\cC_*$ consists of an $\bF_2$-vector space $C_*=\bigoplus_{i\in[t]}C_i$ where each $C_i$ is an $\bF_2$-vector space, along with a linear \textbf{boundary map} $\partial:C_*\rightarrow C_*$ satisfying $\partial^2=0$ and $\partial(C_i)\subseteq C_{i-1}$. Letting $\partial_i=\partial|_{C_i}$, we summarize the components of the chain complex by writing
  \begin{equation*}
    \cC_* = \left(C_{t-1} \xrightarrow{\partial_{t-1}} C_{t-2} \xrightarrow{\partial_{t-2}} \cdots \xrightarrow{\partial_1} C_0\right).
  \end{equation*}

  In this paper we asume all $C_i$ are finite-dimensional. Then $C_*$ has an associated dual chain complex~$\cC^*=(C^*,\delta)$, called the \textbf{cochain complex}, given by
  \begin{equation*}
    \cC^* = \left(C^{t-1} \xleftarrow{\delta_{t-2}} C^{t-2} \xleftarrow{\delta_{t-3}} \cdots \xleftarrow{\delta_0} C^0\right),
  \end{equation*}
  where each $C^i=C_i$, so that $C^*=C_*$,\footnote{In general $C^i$ is the space of linear functions from $C_i$ to $\bF_2$, which may not equal $C_i$ if $C_i$ is infinite-dimensional. However, in this paper we always take all $C_i$ to be finite-dimensional, so $C^i\cong C_i$ and therefore $C^*\cong C_*$.} and $\delta_i=\partial_{i+1}^\top$, so that $\delta=\partial^\top$. The boundary map $\delta$ of $\cC^*$ is called the \textbf{coboundary map} of $\cC_*$.

  For each $i\in[t]$, we furthermore let
  \begin{align*}
    Z_i(\cC) &= \ker\partial_i \hspace{1em}\text{ be the space of \textbf{$i$-cycles}},\\
    B_i(\cC) &= \im\partial_{i+1} \hspace{1em}\text{ be the space of \textbf{$i$-boundaries}}, \\
    H_i(\cC) &= Z_i(\cC)/B_i(\cC) \hspace{1em}\text{ be the \textbf{$i$-homology group}},
  \end{align*}
  where we let $\partial_i=0$ for all $i\notin\{1,\dots,t-1\}$. The \textbf{$i$-cocycles} $Z^i(\cC)$, \textbf{$i$-coboundaries} $B^i(\cC)$, and \textbf{$i$-cohomology group} $H^i(\cC)$ are defined analogously for the cochain complex.

  Sometimes for clarity we will denote the boundary map of $\cC_*$ by $\partial^{\cC}=\partial$, and similarly denote the coboundary map $\delta^{\cC}=\delta$.

  We say that the chain complex $\cC_*$ is \textbf{based} if each $C_i$ has some fixed basis. \end{definition}

For the purpose of this paper, we assume all of our chain complexes are based, and often write ``chain complex'' to implicitly mean ``based chain complex.'' The based condition gives a well-defined notion of Hamming weight:

\begin{definition}
  Let $C$ be a finite-dimension $\bF_2$-vector space with a fixed basis, so that every $c\in C$ can be expressed as $c=(c_0,\dots,c_{n-1})$ where $n=\dim C$ and each $c_i\in\bF_2$. Then the \textbf{Hamming weight} of $c$ is
  \begin{equation*}
    |c| = |\{i\in[n]:c_i\neq 0\}|.
  \end{equation*}
  We extend this definition to sets $S\subseteq C$ (where $S$ may not be a linear subspace) by defining
  \begin{equation*}
    |S| = \min_{c\in S}|c|.
  \end{equation*}
\end{definition}

\subsection{Classical and Quantum Codes as Chain Complexes}
In this section, we present standard definitions of classical and quantum codes. However, we treat these codes using the language of chain complexes, which is standard in the quantum coding literature, but perhaps not in the classical literature.

\begin{definition}
  A \textbf{classical linear code} (or simply classical code) of \textbf{block length} $n$ is a subspace $C\subseteq\bF_2^n$. The \textbf{dimension} $k$ of $C$ is its dimension as an $\bF_2$-vector space, that is, $k=\dim_{\bF_2}C$. The \textbf{distance} $d$ of $C$ is the minimum Hamming weight of a nonzero vector in $C$, that is, $d=\min_{c\in C\setminus\{0\}}|c|$. If $C$ has block length $n$, dimension $k$, and distance $\geq d$, we say that $C$ is an $[n,k,d]$ code. The \textbf{dual} of $C$ is the $(n-k)$-dimensional code $C^\perp=\{c'\in\bF_2^n:c'\cdot c=0\;\forall c\in C\}$.
\end{definition}

It will often be helpful to treat classical codes in the language of chain complexes, as described below.

\begin{definition}
  Every 2-term chain complex $\cC_*=(C_1\xrightarrow{\partial_1}C_0)$ has an \textbf{associated classical code} $\ker\partial_1\subseteq C_1$. The operator $\partial_1:C_1\rightarrow C_0$ is called a \textbf{parity-check matrix} of the code $\ker\partial_1$.
\end{definition}

We now define quantum CSS codes:

\begin{definition}
  A \textbf{quantum CSS code} (or simply quantum code) of \textbf{block length} $n$ is a pair $C=(C_X,C_Z)$ of subspaces $C_X,C_Z\subseteq\bF_2^n$ such that $C_Z^\perp\subseteq C_X$ (and therefore also $C_X^\perp\subseteq C_Z$). The \textbf{dimension} $k$ of $C$ is given by
  \begin{equation*}
    k = \dim(C_Z)-\dim(C_X^\perp) = \dim(C_X)-\dim(C_Z^\perp),
  \end{equation*}
  and the \textbf{distance} $d$ of $C$ is given by
  \begin{equation*}
    d = \min_{c\in (C_Z\setminus C_X^\perp)\cup(C_X\setminus C_Z^\perp)}|c|.
  \end{equation*}
  If $C$ has block length $n$, dimension $k$, and distance $\geq d$, we say that $C$ is an $[[n,k,d]]$ code.
\end{definition}

Just as a classical code can be obtained from a 2-term chain complex, a quantum code can similarly be obtained from a 3-term chain complex:

\begin{definition}
  Every 3-term chain complex $\cC_*=(C_2\xrightarrow{\partial_2}C_1\xrightarrow{\partial_1}C_0)$ has an \textbf{associated quantum code} $C=(C_X=\ker\delta_1,C_Z=\ker\partial_1)$, where we recall that $\delta_1=\partial_2^\top$. The operators $\delta_1$ and $\partial_1$ are called \textbf{X} and \textbf{Z parity-check matrices}, respectively.
\end{definition}

For a 3-term chain complex $\cC_*$, the condition that $\partial_1\partial_2=0$ is equivalent to the condition that the associated CSS code $(C_X,C_Z)$ satisfies $C_X^\perp\subseteq C_Z$. Thus chain complexes indeed are a natural language for describing CSS codes.

By definition, the quantum code associated to a 3-term chain complex $\cC_*$ has dimension
\begin{equation*}
  k = \dim H_1(\cC) = \dim H^1(\cC),
\end{equation*}
and has distance
\begin{equation*}
  d = \min_{c\in(Z_1(\cC)\setminus B_1(\cC))\cup(Z^1(\cC)\setminus B^1(\cC))}|c|.
\end{equation*}

When clear from context, we will sometimes refer to a 2-term chain complex and its associated classical code interchangeably, and we will similarly refer to a 3-term chain complex and its associated quantum code interchangeably.

In this paper we are specifically interested in LDPC codes, which have sparse parity-check matrices, as defined below.

\begin{definition}
  The \textbf{locality} of a chain complex $\cC_*$ is the maximum Hamming weight of any row or column of any boundary map $\partial_i$ in the complex.

  In particular, we say that a family of complexes with locality $\leq w$ is \textbf{LDPC (low-density parity-check) with locality $w$}. If $w=O(1)$ is a constant, we simply say the family is LDPC.
\end{definition}

Thus a family of 2-term chain complexes with sparse boundary maps corresponds to a family of classical LDPC codes, while a family of 3-term chain complexes with sparse boundary maps corresponds to a family of quantum LDPC codes.

\subsection{Expansion and Decoding of Classical Codes}
\label{sec:expdecclass}
We now define some relevant notions of expansion and decoding for classical codes, or rather, for 2-term chain complexes. From this point on the reader may think of most (classical as well as quantum) codes we discuss as being LDPC unless explicitly stated otherwise, though we will only formally impose the LDPC condition when necessary or useful.

To begin, we define expansion for classical codes.

\begin{definition}
  \label{def:ccexp}
  Let $\cC_*=(C_1\xrightarrow{\partial_1}C_0)$ be a 2-term chain complex with $N_i=\dim C_i$. Then $\cC$ has \textbf{$(\alpha,\beta)$-expansion} if for every $c\in C_1$ such that $0<|c|\leq\alpha N_1$, it holds that $|\partial_1c|\geq\beta|c|$.
\end{definition}

\begin{remark}
  Perhaps more descriptive term than \textit{expansion} would be \textit{small-set expansion}. There is also a generalization to chain complexes with more terms called \textit{small-set (co)boundary expansion}, which turns out to be important in applications of quantum LDPC codes \cite{hopkins_explicit_2022-1,anshu_nlts_2023}. However, we will not need this notion of small-set (co)boundary expansion for our work.
\end{remark}

Known constructions of classical LDPC codes (e.g.~\cite{sipser_expander_1996}) have $(\Omega(1),\Omega(1))$-expansion for both the associated chain complex $\cC_*$ and its cochain complex $\cC^*$ (see Lemma~\ref{lem:ctexp} below).

Many of these expanding classical codes have the additional property that they can be approximately decoded from noisy syndromes, as defined below.

\begin{definition}
  Let $\cC_*=(C_1\xrightarrow{\partial_1}C_0)$ be a 2-term chain complex with $N_i=\dim C_i$. We say that (the code associated to) $\cC_*$ is \textbf{$(e_0,e_1,\gamma)$-noisy-syndrome decodable in time $T$} if there exists a decoding algorithm $D:C_0\rightarrow C_1$ (specifying a possibly nonlinear function) that runs in time $T$ such that for every code error $c_1\in C_1$ of weight $|c_1|\leq e_1$ and every syndrome error $c_0\in C_0$ of weight $|c_0|\leq e_0$, it holds that
  \begin{equation*}
    |D(c_0+\partial_1c_1)-c_1|\leq\gamma|c_0|.
  \end{equation*}
\end{definition}

In words, noisy-syndrome decodability says that given code error $c_1$ and syndrome error $c_0$, a decoder receiving the noisy syndrome $c_0+\partial c_1$ is able to approximately recover the code error, up to a loss that is at most proportional to the weight of the syndrome error $c_0$. Note that this loss is bounded independently of the size of the code error $c_1$, as long as $|c_1|\leq e_1$. For instance, in the special case of a noiseless syndrome $c_0=0$, whenever $|c_1|\leq e_1$ the decoder exactly recovers $c_1$.

This notion of noisy-syndrome decodability is really a property of the chain complex $\cC_*$ rather than of the code $\ker\partial_1$. Indeed, the definition is meaningful even when $\ker\partial_1=\{0\}$, and we will use the definition for such chain complexes.

\begin{remark}
  \label{remark:noisysyn}
  Many asymptotically good $[N,\Theta(N),\Theta(N)]$ classical LDPC codes in the literature that are based on expander graphs (e.g.~those in \cite{sipser_expander_1996}) are $(\Theta(N),\Theta(N),\Theta(1))$-noisy-syndrome decodable in $O(N)$ time.\footnote{An exception is given by LDPC codes based in \textit{unique-neighbor expanders}, which often have no known efficient decoders.} Such a result on noisy-syndrome decoding is for instance shown in \cite{spielman_linear-time_1996}. However, for our purposes, we will need $\cC$ for which both the chain complex $\cC_*$ and its cochain complex $\cC^*$ are noisy-syndrome decodable in linear time. We could not find such a result written in the literature, so we prove the existence of classical LDPC codes $\cC$ with such noisy-syndrome decoders in Proposition~\ref{prop:classtan}.

  Decoding from noisy syndromes in the quantum setting is often referred to as \textit{single-shot decoding} \cite{bombin_single-shot_2015} (see also e.g.~\cite{campbell_theory_2019-1,fawzi_constant_2020,kubica_single-shot_2022,gu_single-shot_2023}). In this paper, we use classical noisy-syndrome decoding to construct decoders for quantum codes, but we do not consider noisy syndromes of the quantum codes themselves. It is an interesting question to determine if our techniques extend to this setting.
\end{remark}



\subsection{Decoding of Quantum Codes}
We now define the decoding problem for quantum codes. Unlike the classical treatment in Section~\ref{sec:expdecclass}, for quantum decoding we assume the decoder has access to an exact noise-free error syndrome. It is an interesting direction for future work to extend our results to allow for noisy syndromes in the quantum codes.



\begin{definition}
  \label{def:decoder}
  Let $\cC_*=(C_2\xrightarrow{\partial_2}C_1\xrightarrow{\partial_1}C_0)$ be a 3-term chain complex. A \textbf{decoder against $e$ (adversarial) errors} for the quantum code associated to $\cC_*$ is a pair of algorithms $D_0:C_0\rightarrow C_1$ and $D^2:C^2\rightarrow C^1$ (specifying possibly nonlinear functions) with the following properties:
  \begin{enumerate}
  \item For every $c\in C_1$ with $|c|\leq e$, it holds that $D_0(\partial_1c)\in c+B_1(\cC)$.
  \item For every $c\in C^1$ with $|c|\leq e$, it holds that $D^2(\delta_1c)\in c+B^1(\cC)$.
  \end{enumerate}
  If $D_0,D^2$ are randomized algorithms and the above statements hold (for every $c$) with probability $\geq 1-\epsilon$ over the decoder's randomness, we say that the decoder is \textbf{randomized} with \textbf{failure probability $\epsilon$}.
  
  The \textbf{running time} of the decoder for $\cC_*$ is simply the sum of the running times of $D_0$ and $D^2$.
\end{definition}

In words, Definition~\ref{def:decoder} says that a decoder for a CSS code simply consists of an $X$ and a $Z$ decoder, each of which is given access to a sufficiently low-weight $X$ or $Z$ error syndrome respectively, and must recover the original error up to (physically irrelevant) added stabilizers.

We are interested in finding decoders that decode against a close-to-optimal number of errors. The optimal number of adversarial errors that can be corrected is one less than half the code's distance, so we typically want decoders protecting against a linear number of errors with respect to the distance.


\subsection{Hypergraph Product}
\label{sec:hgpdef}
This section describes how the hypergraph (i.e.~homological) product of two 2-term chain complexes representing classical codes yields a 3-term chain complex representing a quantum code. This construction and its distance analysis was given in \cite{tillich_quantum_2014}.

\begin{definition}
  \label{def:hgpprod}
  Let $\cA_*=(A_1\xrightarrow{\partial^{\cA}}A_0)$ and $\cB_*=(B_1\xrightarrow{\partial^{\cB}}B_0)$ be 2-term chain complexes. Then the \textbf{hypergraph product} $(\cA\otimes_{\bF_2}\cB)_*$ is the 3-term chain complex
  \begin{equation*}
    A_1\otimes B_1 \xrightarrow{\partial^{\cA}\otimes I+I\otimes\partial^{\cB}} A_0\otimes B_1\oplus A_1\otimes B_0 \xrightarrow{I\otimes\partial^{\cB}+\partial^{\cA}\otimes I} A_0\otimes B_0.
  \end{equation*}
  All tensor products above are taken over $\bF_2$, so we often omit the $\bF_2$ subscript and for instance write $\cA\otimes\cB$.
\end{definition}

The following formula for the dimension of a hypergraph product code was shown in \cite{tillich_quantum_2014}, and also follows from the well-known K\"{u}nneth formula (see e.g.~\cite{hatcher_algebraic_2001}).

\begin{proposition}[K\"{u}nneth formula]
  \label{prop:kunneth}
  For 2-term chain complexes $\cA_*$ and $\cB_*$, it holds that
  \begin{equation*}
    \dim H_1(\cA\otimes\cB) = \dim(H_1(\cA))\cdot\dim(H_0(\cB))+\dim(H_0(\cA))\cdot\dim(H_1(\cB)).
  \end{equation*}
\end{proposition}

\cite{tillich_quantum_2014} also showed the following bound on the distance of a hypergraph product code.

\begin{proposition}[\cite{tillich_quantum_2014}]
  \label{prop:hgpdistance}
  For 2-term chain complexes $\cA_*$ and $\cB_*$, the quantum code associated to the 3-term complex $\cA\otimes\cB$ has distance
  \begin{align*}
    d(\cA\otimes\cB) &\geq \min\{d(\cA_*),d(\cB_*),d(\cA^*),d(\cB^*)\},
  \end{align*}
  where for a 2-term chain complex $\cC_*$, we let $d(\cC_*)$ denote the distance of the associated classical code.
\end{proposition}

\cite{tillich_quantum_2014} instantiated the hypergraph product construction by choosing $\cA_*$ and $\cB^*$ to be 2-term chain complexes associated to asymptotically good $[N,\Theta(N),\Theta(N)]$ classical LDPC codes, such as lossless expander codes or Tanner codes (see e.g.~\cite{sipser_expander_1996}). In this case, the hypergraph product $\cC=\cA\otimes\cB$ corresponds to a $[[\Theta(N^2),\Theta(N^2),\Theta(N)]]$ quantum LDPC code by Proposition~\ref{prop:hgpdistance}. Such products $\cC$ permit linear time ``flip''-style decoders \cite{leverrier_quantum_2015,fawzi_efficient_2018,fawzi_constant_2020}, which greedily flip code bits in an appropriate manner to reduce the syndrome weight.

\subsection{Lifted Product over Cyclic Groups}
\label{sec:lpdef}
This section describes the lifted product, which is a generalized version of the hypergraph product involving a group symmetry. Such generalized products were used in a recent line of work \cite{hastings_fiber_2021,panteleev_quantum_2022,breuckmann_balanced_2021} culminating in the construction of asymptotically good qLDPC codes \cite{panteleev_asymptotically_2022,leverrier_quantum_2022-1,dinur_good_2023}.

Our presentation of lifted products is similar to \cite{panteleev_quantum_2022,panteleev_asymptotically_2022}. In this paper, we focus on the codes of \cite{panteleev_quantum_2022}, which use lifted products over cyclic groups. Hence we will restrict attention in our presentation of the lifted product to such cyclic groups. We first define the relevant group algebra over these groups.

\begin{definition}
  \label{def:groupalg}
  Let $R_\ell=\bF_2[X]/(X^\ell-1)\cong\bF_2[\bZ/\ell\bZ]$ denote the \textbf{group algebra over $\bF_2$ of the cyclic group $\bZ/\ell\bZ$ of order $\ell$}. For an element $f(X)\in R_\ell$, we define the \textbf{conjugate} element $f(X)^*=f(X^{\ell-1})$.
\end{definition}

We emphasize that in this paper, $X$ denotes an indeterminate variable in polynomials over $\bF_2$, \textit{not} a Pauli operator.

Because we restrict attention to the cyclic group, which is abelian, the associated group algebra $R_\ell$ is commutative. Lifted products can also be defined for non-abelian groups, as was used in \cite{breuckmann_balanced_2021,panteleev_asymptotically_2022}. However, restricting to the abelian case will slightly simplify our presentation.

\begin{definition}
  A \textbf{$t$-term chain complex over $R_\ell$} is a $t$-term chain complex (over $\bF_2$) $\cC_*$ with the additional property that each $C_i$ is a free $R_\ell$-module with a fixed basis, and the boundary map $\partial:C_*\rightarrow C_*$ is an $R$-module homomorphism.
\end{definition}

In this paper we restrict attention to finite-dimensional complexes, so for each $i\in[t]$ we have $C_i\cong R_\ell^{n_i}$, and the $i$th boundary map $\partial_i:R_\ell^{n_i}\rightarrow R_\ell^{n_{i-1}}$ may be expressed as an $n_{i-1}\times n_i$ matrix of elements of $R_\ell$.

As $R_\ell=\bF_2[X]/(X^\ell-1)$ has the natural basis $\{1,X,X^2,\dots,X^{\ell-1}\}$, each $R_\ell$-module $C_i\cong R_\ell^{n_i}$ with a given $R_\ell$-basis then has an induced $\bF_2$-basis. Thus a based chain complex over $R_\ell$ is always a based chain complex over $\bF_2$. Note that we use this basis over $\bF_2$ to compute Hamming weights, locality, cochain complexes, and so forth with respect to the $\bF_2$-structure. For instance, we compute the Hamming weight of some $c\in C_i$ by counting the number of nonzero coefficients in the decomposition of $c$ into the $\bF_2$-basis, not in the $R_\ell$-basis.

Similarly, cochain complexes are also still defined as in Definition~\ref{def:chaincomplex} using the $\bF_2$-structure, even when we have an $R_\ell$-structure. In particular, consider a chain complex $\cC_*$ over $R_\ell$ with a boundary map $\partial_i:R_\ell^{n_i}\rightarrow R_\ell^{n_{i-1}}$ given by a matrix $\partial_i=H=(H_{jk})_{j\in[n_{i-1}],k\in[n_i]}\in R_\ell^{n_{i-1}\times n_i}$. Then the associated cochain complex $\cC^*$ has coboundary map $\delta_{i-1}:R_\ell^{n_{i-1}}\rightarrow R_\ell^{n_i}$ given by the conjugate transpose matrix of $\partial_i$, that is, $\delta_i=\partial_i^\dagger=(H_{kj}^*)_{k\in[n_i],j\in[n_{i-1}]}$, using the notion of conjugate in Definition~\ref{def:groupalg}.

We are now ready to define a lifted product.

\begin{definition}
  Let $\cA_*=(A_1\xrightarrow{\partial^{\cA}}A_0)$ and $\cB_*=(B_1\xrightarrow{\partial^{\cB}}B_0)$ be 2-term chain complexes over $R_\ell\cong\bF_2[\bZ/\ell\bZ]$. Then the \textbf{$\bZ/\ell\bZ$-lifted product $(\cA\otimes_{R_\ell}\cB)_*$} (or simply ``$\ell$-lifted product'' or ``lifted product'' for short) is the 3-term chain complex
  \begin{equation*}
    A_1\otimes_{R_\ell} B_1 \xrightarrow{\partial^{\cA}\otimes_{R_\ell} I+I\otimes_{R_\ell}\partial^{\cB}} A_0\otimes_{R_\ell} B_1\oplus A_1\otimes_{R_\ell} B_0 \xrightarrow{I\otimes_{R_\ell}\partial^{\cB}+\partial^{\cA}\otimes_{R_\ell} I} A_0\otimes_{R_\ell} B_0.
  \end{equation*}
\end{definition}

Assuming $\cA_*$ and $\cB_*$ have specified bases, then the lifted product $\cA\otimes_{R_\ell}\cB$ has an induced basis. This statement follows from the general fact that if $A$ and $B$ are free $R_\ell$-modules with specified $R_\ell$-bases, then $A\otimes_{R_\ell}B$ is a free $R_\ell$-module with basis elements $a\otimes_{R_\ell}b$ for all $R_\ell$-basis elements $a\in A$ and $b\in B$.

We remark that $\otimes_{\bF_2}$ and $\otimes_{R_\ell}$ are in general non-equivalent operators. When no subscript is specified we assume $\otimes$ refers to $\otimes_{\bF_2}$, unless it is clear from context or explicitly specified that we are referring to $\otimes_{R_\ell}$.

The main lifted products we consider are those described in the following result of \cite{panteleev_quantum_2022}. To state this result, we need to translate repetition codes to the language of chain complexes. For $\ell\in\bN$, we define the \textbf{chain complex $\cB_*=(B_1\xrightarrow{\partial^{\cB}}B_0)$ of a length-$\ell$ repetition code} to be the chain complex over $R_\ell$ given by $B_0=B_1=R_\ell$ and $\partial^{\cB}=1+X$. Equivalently, as a chain complex over $\bF_2$, we have $B_0=B_1=\bF_2^\ell$ and $\partial^{\cB}\1_i=\1_i+\1_{i+1\pmod{\ell}}$ for $i\in[\ell]$. Note that indeed $\cB_*$ and its cochain complex $\cB^*$ both have associated code $\ker\partial^{\cB}=\ker\delta^{\cB}$ equal to the length-$\ell$ repetition code.

\begin{theorem}[\cite{panteleev_quantum_2022}]
  \label{thm:pkdisinf}
  \label{thm:pkdis}
  Let $\cA^{(\ell)},\ell\in\bN$ denote a family of explicit $(\Omega(1),\Omega(1))$-expanding\footnote{Meaning that both $\cA^{(\ell)}_*$ and $\cA^{(\ell)}_*$ are $(\Omega(1),\Omega(1))$-expanding in the sense of Definition~\ref{def:ccexp}.} 2-term chain complexes over $R_\ell$ of constant locality such that $n_0,n_1,n_1-n_0=\Theta(\log\ell)$ for $n_i:=\dim_{R_\ell}A_i^{(\ell)}$, as for instance given in Proposition~\ref{prop:classtan}. Let $\cB^{(\ell)}=(R_\ell\xrightarrow{1+X}R_\ell)$ be the repetition code complex. Then the 3-term chain complexes $\cC^{(\ell)}:=\cA^{(\ell)}\otimes_{R_\ell}\cB^{(\ell)}$ form an explicit family of $[[\Theta(\ell\log\ell),\Theta(\log\ell),\Theta(\ell))]]$ quantum LDPC codes.
\end{theorem}

Chain complexes $\cA^{(\ell)}$ with the necessary properties for Theorem~\ref{thm:pkdisinf} to hold are provided in Proposition~\ref{prop:classtan}. These complexes are constructed as classical Tanner codes on a $\Theta(\ell\log\ell)$-vertex expander graph that respects a group action of $\bZ/\ell\bZ$.

\section{Technical Overview of Decoding Algorithms}
\label{sec:tecoverview}
This section provides an overview of our decoding algorithms. Section~\ref{sec:hgpdecinf} describes our hypergraph product decoder, which serves as a good warm-up to Section~\ref{sec:lpdecinf}, where we describe our lifted-product decoder that we use to prove our main result, namely Theorem~\ref{thm:lpdecinf}. The full descriptions and proofs of these results can be found in Section~\ref{sec:hgpdec} for the hypergraph product, and Section~\ref{sec:lpdec} for the lifted product.

\subsection{Decoding Hypergraph Products}
\label{sec:hgpdecinf}
In this section, we describe our decoding algorithm for the hypergraph product of an expander-based classical LDPC code with a repetition code. Specifically, we will prove Theorem~\ref{thm:hgpdecinf} below. We remark that in Appendix~\ref{sec:spacetime}, we give an alternative, fully classical view of this decoding algorithm by interpreting the repetition code as a time axis. Similar, but more involved, techniques will be used in our lifted product decoder described in Section~\ref{sec:lpdecinf} below.

\begin{theorem}[Informal statement of Theorem~\ref{thm:hgpdec}]
  \label{thm:hgpdecinf}
  Let $\cA_*=(A_1\xrightarrow{\partial^{\cA}}A_0)$ be a 2-term chain complex of constant locality such that $\cA_*$ and $\cA^*$ are both $(\Theta(N^{\cA}),\Theta(N^{\cA}),\Theta(1))$-noisy-syndrome decodable in time $O(N^{\cA})$, where $N^{\cA}:=\dim A_0+\dim A_1$. For some $\ell\in\bN$, let $\cB_*=(B_1\xrightarrow{\partial^{\cB}}B_0)$ be the 2-term chain complex of a length-$\ell$ repetition code. Let
  \begin{equation*}
    \cC_* = \cA\otimes_{\bF_2}\cB
  \end{equation*}
  be the hypergraph product. Then for every $\delta>0$, there exists a randomized decoder for this quantum code that protects against $e=\Theta(\min\{N^{\cA},\ell\})$ adversarial errors, runs in time $O(N^{\cA}\ell\log(1/\delta))$, and has failure probability $\leq\delta$. Furthermore, there exists a deterministic decoder protecting against $e$ errors that runs in time $O(N^{\cA}\ell^2)$.
\end{theorem}

If we let $\cA_*$ be a $[N,\Theta(N),\Theta(N)]$ classical expander-based LDPC code and we set $\ell=\Theta(N)$, then $\cC=\cA\otimes\cB$ is the hypergraph product of $\cA$ with the length-$\Theta(N)$ repetition code $\cB$, so $\cC$ yields a $[[\Theta(N^2),\Theta(N),\Theta(N)]]$ quantum LDPC code. Theorem~\ref{thm:hgpdecinf} provides a randomized decoder for this code $\cC$ with almost linear running time with respect to the block length, that corrects arbitrary adversarial errors of weight linear in the code distance with negligible failure probability. Theorem~\ref{thm:hgpdecinf} also provides a deterministic decoder (with zero probability of failure) that corrects the same number of errors in time $O(N^2)^{3/2}$, where $\Theta(N^2)$ is the block length.

This $[[\Theta(N^2),\Theta(N),\Theta(N)]]$ qLDPC code $\cC$ has worse dimension than the $[[\Theta(N^2),\Theta(N^2),\Theta(N)]]$ codes constructed as products of two expander-based codes that were described above in Section~\ref{sec:hgpdef}. However, as described in Section~\ref{sec:introconstruct}, codes such as $\cC$ constructed as products of an expander-based code with a repetition code arise naturally in certain settings. Theorem~\ref{thm:hgpdecinf} also serves as a good warm-up to our decoding algorithm for the \textit{lifted product} of such codes, which is more involved, as we will subsequently describe.

Our algorithm in Theorem~\ref{thm:hgpdecinf} differs from the ``flip''-style decoders used for products of two expander-based LDPC codes. The parity-check matrix $\partial^{\cB}$ of the repetition code has sufficiently poor expansion (its Tanner graph is a single long cycle) that we are not able to use a flip-style decoder $\cC$. Instead, our algorithm leverages the symmetry of the repetition code to compute prefix sums of the syndrome. Our algorithm then passes these prefix sums (in a black-box fashion) through a classical decoder that is robust to errors on both code bits and syndrome bits. A sketch of our algorithm and analysis is provided below, with an accompanying illustration in Figure~\ref{fig:algpic}; for the full details, the reader is referred to Section~\ref{sec:hgpdec}.

\begin{figure}
  \centering
  \includegraphics[width=6in]{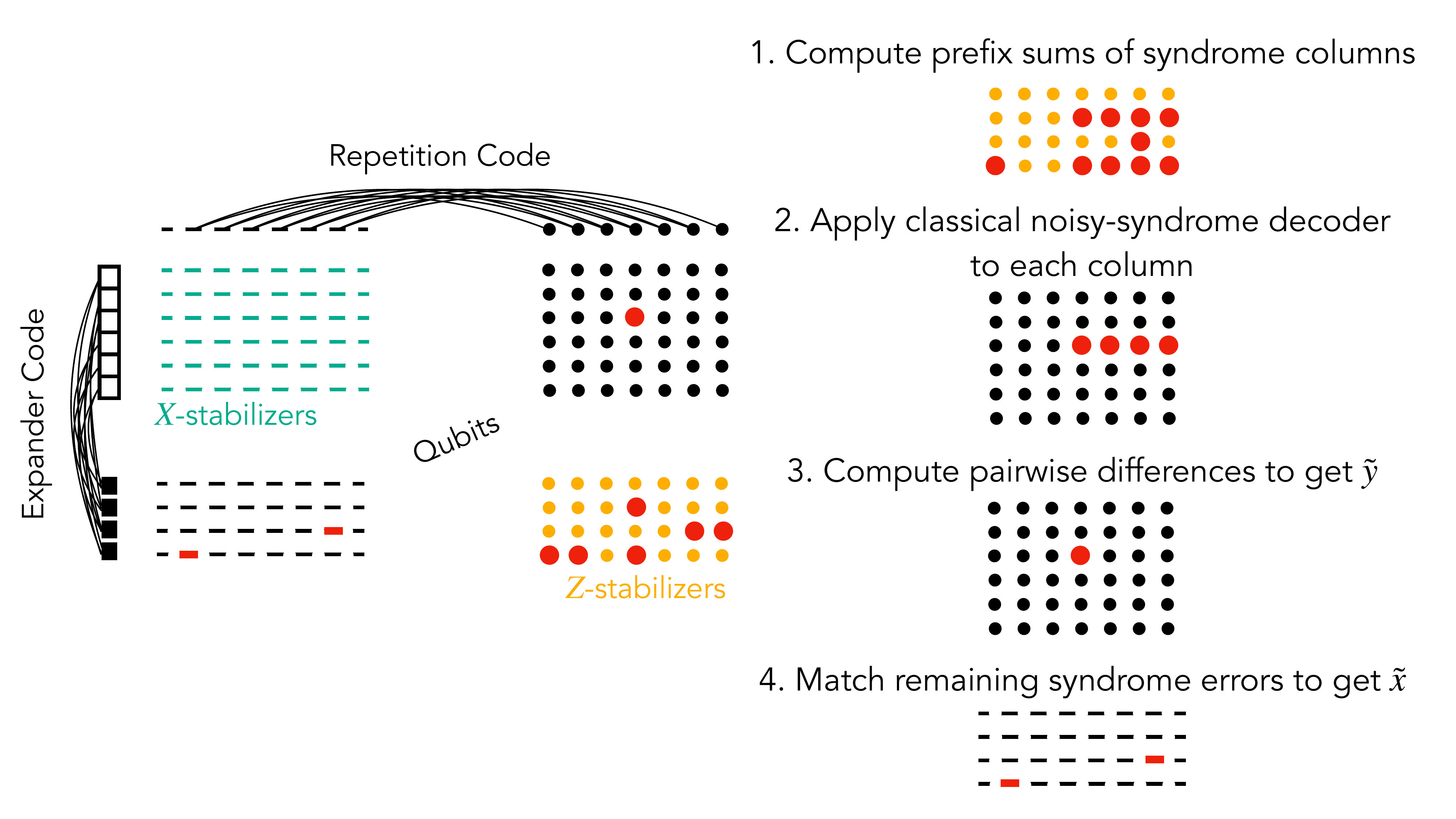}
  \caption{\label{fig:algpic} Illustration of our decoding algorithm in Theorem~\ref{thm:hgpdecinf} for the hypergraph product of a classical expander-based LDPC code with a repetition code. The hypergraph product is drawn on the left; there are two matrices of qubits, corresponding to $x$ and $y$ respectively in the proof sketch of Theorem~\ref{thm:hgpdecinf}, where the errors are supported on the qubits marked by red dots. The support of the $Z$-syndrome $s=\partial_1^{\cC}(x,y)=(I\otimes\partial^{\cB})x+(\partial^{\cA}\otimes I)y$ is also marked by red dots within the matrix of $Z$-stabilizers (i.e.~parity checks of the boundary map $\partial_1^{\cC}$). The algorithm outputs an estimate $(\tilde{x},\tilde{y})$ for the true error $(x,y)$.}
\end{figure}

\begin{proof}[Proof sketch of Theorem~\ref{thm:hgpdecinf}]
  By the symmetry of $\cC_*$ and its cochain complex $\cC^*$ (see Definition~\ref{def:hgpprod}), it suffices to describe the decoder $D_0:C_0\rightarrow C_1$ as described in Definition~\ref{def:decoder}; the construction of $D^2:C^2\rightarrow C^1$ will be exactly analogous. Thus, if we fix some error $c=(x,y)\in C_1=A_0\otimes\bF_2^\ell\oplus A_1\otimes\bF_2^\ell$, our goal is to construct an efficient algorithm $D_0(s)$ that receives as input the syndrome
  \begin{equation}
    \label{eq:hgpsyninf}
    s = \partial_1^{\cC}(c) = (I\otimes\partial^{\cB})x+(\partial^{\cA}\otimes I)y,
  \end{equation}
  and outputs some $\tilde{c}\in C_1$ that differs from $c$ by an element $\tilde{c}-c\in\im\partial_2^{\cC}$. In fact, we can define $e=\Theta(\min\{N^{\cA},\ell\})$ to be sufficiently small compared to the distance of $\cC$ so that it suffices to output any $\tilde{c}$ of weight $|\tilde{c}|\leq O(|c|)$ such that $\partial_1^{\cC}\tilde{c}=s$.

  For $i\in\{0,1\}$, let $N_i=\dim A_i$, so that $A_i\otimes\bF_2^\ell\cong\bF_2^{N_i\times\ell}$ can be viewed as the space of $N_i\times\ell$ matrices over $\bF_2$.

  First observe that if $x=0$, then $s=(\partial^{\cA}\otimes I)y$ is simply the matrix whose columns are syndromes of the columns of $y$ under the parity-check matrix $\partial^{\cA}$ of $\cA$. By the assumption that $\cA$ has an efficient noisy-syndrome decoder $D^{\cA}:A_0\rightarrow A_1$, so we can simply apply $D^{\cA}$ to each column of $s$ independently to recover $y$. Note that here we did not even need the ability of $D^{\cA}$ to accomodate noisy syndromes.

  Now in the general case where $x\neq 0$, then we can interpret $s=(I\otimes\partial^{\cB})x+(\partial^{\cA}\otimes I)y$ as being the matrix whose columns are \textit{noisy} syndromes of the columns of $y$ under the parity-check matrix $\partial^{\cA}$, where $(I\otimes\partial^{\cB})x$ provides the syndrome errors. Letting $a_i$ denote the $i$th column of a matrix $a$, then observe that $((I\otimes\partial^{\cB})x)_i=x_{i-1}+x_i$, where we view $[\ell]\cong\bZ/\ell\bZ$ so that $i-1$ is taken$\pmod{\ell}$.

  Therefore an initial attempt at decoding would be to again apply $D^{\cA}$ to each column of $s$. This would yield a matrix with columns $D^{\cA}(s_i)$ satisfying $|D^{\cA}(s_i)-y_i|\leq O(|x_{i-1}+x_i|)$ by the definition of noisy-syndrome decodability. However, it is not clear that we have made any progress, as we simply know that the matrix with columns $D^{\cA}(s_i)$ differs from $y$ in at most $\sum_iO(|x_{i-1}+x_i|)\leq O(|x|)$ positions. Because we could have $|x|=\Theta(|y|)=\Theta(e)$, the all-zeros matrix could provide an equally good approximation to $y$.

  Thus we need a way to isolate, or amplify, the effect of the errors $y$ while suppressing the syndrome errors from $x$, so that we can gain more information from our noisy-syndrome decoder $D^{\cA}$. The solution is to apply $D^{\cA}$ to prefix sums of $s$. That is, for every $k\in[\ell+1]$, we compute
  \begin{equation}
    \label{eq:tadefinf}
    \tilde{a}_k = D^{\cA}\Bigl(\sum_{i\in[k]}s_i\Bigr),
  \end{equation}
  which we hope to be a good approximation to the true prefix sums
$a_k = \sum_{i\in[k]}y_i$. 

  The key observation here is that while passing to these prefix sums potentially allows an error in some column $y_i$ to propagate to all $a_k$ for $k>i$, the contribution of errors in $x$ to prefix sums of the syndrome $s$ is not amplified. Formally, by definition $\sum_{i\in[k]}((I\otimes\partial^{\cB})x)_i=\sum_{i\in[k]}(x_{i-1}+x_i)=x_{-1}+x_{k-1}$ (where $x_{-1}=x_{-1\pmod{\ell}}=x_{\ell-1}$). Thus if $x_{-1}=0$, then $x$ contributes at most $|x|$ nonzero entries to the matrix whose $k$th column is the prefix sum $\sum_{i\in[k]}s_i$. We can guarantee that $x_{-1}=0$ by permuting the columns of all matrices by an appropriate (e.g.~random) cyclic shift. We obtain the randomized decoder in Theorem~\ref{thm:hgpdecinf} by trying a small number of random cyclic shifts, while we obtain the deterministic decoder by trying all $\ell$ possible cyclic shifts.

  Thus by the definition of noisy-syndrome decoding, it must hold that
  \begin{equation}
    \label{eq:taainf}
    \sum_{k\in[\ell+1]}|\tilde{a}_k-a_k| \leq O(|x|).
  \end{equation}
  Note that here we are able to apply $(\Theta(n),\Theta(n),\Theta(1))$-noisy-syndrome decoding because $|x|$ (which bounds the syndrome errors) and each $|a_k|\leq|y|$ (which bounds the code bit errors) are by assumption bounded above by $e\leq\Theta(n)$.

  Therefore after computing $\tilde{a}_k$ for every $k\in[\ell+1]$, our decoder computes an estimate $\tilde{y}\in\bF_2^{N_1\times\ell}$ by letting the $k$th column be $\tilde{y}_k:=\tilde{a}_{k+1}-\tilde{a}_k$. Tracing through the definitions, we see that $\tilde{y}-y=(I\otimes\partial^{\cB})b$ for the matrix $b\in\bF_2^{N_1\times\ell}$ given by $b_k:=\tilde{a}_{k+1}-a_{k+1}$, as for every $k$ we have $(\tilde{y}-y)_k=b_{k-1}+b_{k}$.

  Therefore we have shown that $\tilde{y}=y+(I\otimes\partial^{\cB})b$ for a matrix $b$ satisfying $|b|\leq O(|x|)$ by~(\ref{eq:taainf}). Thus our decoder simply computes the minimum possible weight $\tilde{x}\in\bF_2^{N_0\times\ell}$ such that $\partial_1^{\cC}(\tilde{x},\tilde{y})=s$, or equivalently, such that $(I\otimes\partial^{\cB})\tilde{x}=s-(\partial^{\cA}\otimes I)\tilde{y}$. It then outputs $\tilde{c}:=(\tilde{x},\tilde{y})$. As $\tilde{x}'=x-(\partial^{\cA}\otimes I)b$ is a valid such choice of $\tilde{x}$, the decoder is guaranteed to output $\tilde{c}$ of weight $|\tilde{c}|=|\tilde{x}|+|\tilde{y}|\leq O(|x|+|y|)$, as desired.
\end{proof}

\subsection{Decoding Lifted Products}
\label{sec:lpdecinf}
In this section, we outline the main ideas in our proof of Theorem~\ref{thm:lpdecinf}, which provides a decoding algorithm for the nearly-linear-distance lifted-product codes in Theorem~\ref{thm:pkdisinf}. This algorithm is one of the main technical contributions of our paper.

We first reiterate Theorem~\ref{thm:lpdecinf} in slightly more technical terms:

\begin{theorem}[Informal statement of Theorem~\ref{thm:lpdec}]
  \label{thm:lpdecinf2}
  For some $\ell\in\bN$ that is a power of $2$, define $\cA^{(\ell)}$, $\cB^{(\ell)}$, and $\cC^{(\ell)}$ as in Theorem~\ref{thm:pkdisinf}. Let $N=\ell\log\ell$, and assume that $\cA^{(\ell)}_*$ and its cochain complex ${\cA^{(\ell)}}^*$ are both $(\Theta(N),\Theta(N),\Theta(1))$-noisy-syndrome decodable in time $O(N)$. Then for every $0<\epsilon\leq 1/2$ and every $\delta>0$, there exists a randomized decoder for $\cC^{(\ell)}$ that protects against $e=\Theta(\epsilon\ell)$ adversarial errors with failure probability $\leq\delta$ and runs in time $O(N^{2+2\epsilon}\cdot\log(1/\delta))$.
\end{theorem}

In Theorem~\ref{thm:lpdecinf2}, we assume that $\ell$ is a power of $2$ to avoid rounding issues in one step of the proof, but we believe that the algorithm can be extended to accomodate all $\ell\in\bN$.

As described in Remark~\ref{remark:noisysyn} above, we expect many expander-based $[\Theta(N),\Theta(N),\Theta(N)]$ classical LDPC codes to be $(\Theta(N),\Theta(N),\Theta(1))$-noisy-syndrome decodable, and we prove the precise statement we need in Proposition~\ref{prop:classtan}.

Below, we outline the main ideas in our proof of Theorem~\ref{thm:lpdecinf2}. For the full details, the reader is referred to Section~\ref{sec:lpdec}.

\begin{proof}[Proof sketch of Theorem~\ref{thm:lpdecinf2}]
  At a high level, we follow the idea of applying noisy-syndrome decoding to prefix sums of the received syndrome, as was described in Section~\ref{sec:hgpdecinf} above for the proof of Theorem~\ref{thm:hgpdecinf}. However, the prefix sums in the lifted product case rapidly become too high-weight to apply the noisy-syndrome decoder $D^{\cA}$ of $\cA$. We must instead start by taking prefix sums of constant size, and then iteratively passing larger prefix sums through $D^{\cA}$ to obtain successively better approximations to the true error. Each successive improvement in our error approximation allows us to take an even larger prefix sum without overloading the error tolerance of $D^{\cA}$. We can ultimately arrive at a true decoding after $\Theta(\log\ell)$ of these iterations. However, each iteration has a constant probability of failure, as we must sample a random cyclic shift in like in Theorem~\ref{thm:hgpdecinf}. Thus we obtain an overall success probability of $2^{-\Theta(\log\ell)}=1/\poly(\ell)$, which we amplify by repeating the algorithm $\poly(\ell)$ times.

  We now provide more details. Fix $\ell\in\bN$, and let $\cA=\cA^{(\ell)}$, $\cB=\cB^{(\ell)}$, and
  \begin{equation*}
    \cC_* = \cC_*^{(\ell)} = (\cA\otimes_{R_\ell}\cB)_* = \left(A_1 \xrightarrow{\partial^{\cA}+(1+X)} A_0\oplus A_1 \xrightarrow{(1+X)+\partial^{\cA}} A_0\right).
  \end{equation*}
  By the symmetry of $\cC_*$ and $\cC^*$, it suffices to construct the decoder $D_0:C_0\rightarrow C_1$ described in Definition~\ref{def:decoder}, as the decoder $D^2:C^2\rightarrow C^1$ will be exactly analogous.

  Thus fix some error $c=(x,y)\in C_1=A_0\oplus A_1$. Our goal is to construct an efficient algorithm $D_0(s)$ that receives as input the syndrome
  \begin{equation*}
    s = \partial_1^{\cC}=(1+X)x+\partial^{\cA}y,
  \end{equation*}
  and outputs some $\tilde{c}\in C_1$ for which $\tilde{c}-c\in\im\partial_2^{\cC}$. In fact, we can define $e=\Theta(\epsilon\ell)$ to be sufficiently small compared to the distance $\Theta(\ell)$ of $\cC$ so that it suffices to output any $\tilde{c}$ of weight $|\tilde{c}|\leq O(|c|)$ such that $\partial_1^{\cC}\tilde{c}=s$.

  Similarly to the proof of Theorem~\ref{thm:hgpdecinf}, if $x=0$, then $s=\partial^{\cA}y$, so we can simply apply $D^{\cA}$ to $s$ to recover $y$. Moving to the general case where $x\neq 0$, we can view $(1+X)x$ as a syndrome error and $y$ as a code error, which both contribute to the noisy syndrome $s=(1+X)x+\partial^{\cA}y$.

  A natural approach is to follow our proof of Theorem~\ref{thm:hgpdecinf} and apply $D^{\cA}$ to prefix sums of $s$. Specifically, the natural analogue of~(\ref{eq:tadefinf}) in the lifted product setting is to define $\tilde{a}_k$ for $k\in[\ell+1]$ by
  \begin{equation}
    \label{eq:talpinf}
    \tilde{a}_k = D^{\cA}\Bigl(\sum_{i\in[k]}X^is\Bigr),
  \end{equation}
  which we hope to be a good approximation to
  $a_k := \sum_{i\in[k]}X^iy$.

  To see that multiplying by $\sum_{i\in[k]}X^i$ is the natural lifted-product analogue of the prefix sums for the hypergraph product, observe that multiplying a vector by $\sum_{i\in[k]}X^i$ simply sums together $k$ cyclic shifts of that vector. But the length-$k$ prefix sums considered in Section~\ref{sec:hgpdecinf} can similar be obtained summing together $k$ cyclic shifts of the columns of the relevant matrices. Thus in both cases, the prefix sums are obtained by summing together shifts of the relevant vectors or matrices under the cyclic group action they inherit from the repetition code.

  Now, we see that the prefix sum $\sum_{i\in[k]}X^is$ receives a contribution of $\partial^{\cA}(\sum_{i\in[k]}X^iy)$ from $y$. As $|y|$ can be as large as $e=\Theta(\epsilon\ell)=\Theta(\epsilon N/\log N)$, the prefix sum $\sum_{i\in[k]}X^iy$ can have weight $>\omega(N)$ for any $k>\omega(\log N)/\epsilon$. Thus unless $k\leq O(\log N)/\epsilon$ (it is helpful to think of $\epsilon>0$ as a constant here),  $\sum_{i\in[k]}X^is$ will represent the noisy syndrome of an error that may be too large for $D^{\cA}$ to decode. In contrast, if we were to use a similar algorithm as in the proof of Theorem~\ref{thm:hgpdecinf}, we would want to compute $\tilde{a}_k$ at least up to $k=\Theta(N/\log N)$.

  We resolve this issue by initially only computing $\tilde{a}_k$ in~(\ref{eq:talpinf}) for all $1\leq k\leq t$ for a constant $t=t_0$ (we in fact take $t_0=2$, but for now it may be helpful to think of $t$ as a large constant, or even as a slowly growing function of $N$). We then form an initial estimate $\tilde{z}\in A_1$ for $y$ using these vectors $\tilde{a}_k$ for $1\leq k\leq t$, similarly as in the proof of Theorem~\ref{thm:hgpdecinf}.

  Specifically, letting $A_1=R_\ell^{n_1}$, then we can index the $n_1\cdot\ell$ bits in a vector $v\in A_1$ by pairs $(h,k)\in[n_1]\times[\ell]$, such that $(X^iv)_{h,k}=v_{h,k+i}$, where we take the second subscript modulo $\ell$. Then for $h\in[n_1]$ and $m\in[\ell/t]$, we define
  \begin{equation*}
    \tilde{z}_{h,mt+i} := \begin{cases}
      0,&i=0 \\
      (\tilde{a}_{i+1})_{h,mt}-(\tilde{a}_{i})_{h,mt},&1\leq i\leq t-1.
    \end{cases}
  \end{equation*}

  Assume that we had initially performed a random cyclic shift on all vectors above, by multiplying all vectors by $X^j$ for a random $j\in[\ell]$, to avoid accumulation of errors on the components $(\tilde{a}_i)_{h,mt}$ that we use above. Then using a similar argument as in the proof of Theorem~\ref{thm:hgpdecinf}, we apply the definition of noisy-syndrome decoding to show that with constant probability, $\tilde{z}_{h,mt+i}$ is a good approximation to $y_{h,mt+i}$ across all\footnote{Here we informally say that a vector $u$ is a good approximation to a vector $v$ across all components in a set $S$ if $u$ and $v$ disagree on a small number of components in $S$.} $h\in[n_1]$, $m\in[\ell/t]$, and $1\leq i<t$. However, we do not obtain an estimate of the components $y_{h,mt+i}$ for $i=0$. These ``gaps'' in our estimate $\tilde{z}$ occur periodically in one out of every $t$ components, and arise from the fact that we only computed $t$ prefix sums $\tilde{a}_k$.

  As these gaps in our estimate $\tilde{z}$ occur at a one out of every $t$ of the positions, where we have chosen $t$ to be a constant, we must find a way to ``fill in'' these gaps to obtain a better estimate of $y$. For this purpose, observe that we have good estimates $\tilde{z}_{h,mt+k}$ of $y_{h,mt+k}$ for all $h,m$ and all $1\leq k<t$, and we also have good estimates $(\tilde{a}_t)_{h,mt-k}$ of $(a_t)_{h,mt-k}=\sum_{i\in[t]}y_{h,mt-k+i}$ for all $h,m,k$. Therefore, for $k\in[t]$, we should expect $b_{h,mt-k}:=(a_t)_{h,mt-k}-\sum_{i\in[t]}\tilde{z}_{h,mt-k+i}$ to be a good estimate of $y_{h,mt}$. We thus define $\tilde{r}\in A_1$ by letting $\tilde{r}_{h,mt}$ be the majority value of $b_{h,mt-k}$ across all $k\in[t]$, and letting $\tilde{r}_{h,mt+i}=0$ for all $1\leq i<t$. We then let $\tilde{y}_1:=\tilde{z}+\tilde{r}$ be our estimate of $y$.

  We show that the error $\tilde{y}_1-y$ in our approximation $\tilde{y}_1$ of $y$ can be expressed in the form $\tilde{y}_1-y=(1+X)u+r'$, where $|u|\leq O(\ell)$ and $|r'|\leq O(\ell/t)$ (see the proof of Proposition~\ref{prop:ampcom} for details). Now consider taking a prefix sum of this error $y-\tilde{y}_1$ of any length $k\in[t_1]$ for $t_1=2t=2t_0$. As
  \begin{equation*}
    \Bigl|\sum_{i\in[k]}X^i(y-\tilde{y}_1)\Bigr| = \Bigl|(1+X^{k})u+\sum_{i\in[2k]}X^ir'\Bigr| \leq 2|u|+2t|r'| \leq O(\ell),
  \end{equation*}
  this prefix sum will not overload the noisy-syndrome decoder $D^{\cA}$. Thus we can compute new decodings $\tilde{a}_k':=D^{\cA}(\sum_{i\in[k]}X^i(s-\partial^{\cA}\tilde{y}_1))$ of the length-$k$ prefix sums of $s-\partial^{\cA}\tilde{y}_1$ for all lengths $k\in[t_1]$, and repeat the entire procedure described above to obtain a new estimate $\tilde{y}_2$ for $y$. After iterating $\Theta(\log\ell)$ times, we show that the error term $r'$ described above vanishes, at which point we have recovered some $\tilde{y}$ for which $\tilde{y}-y=(1+X)u$, where $|u|\leq O(\ell)$. Similarly as in the proof of Theorem~\ref{thm:hgpdecinf}, the decoder can then compute the minimum possible weight $\tilde{x}\in A_0$ for which $\partial_1^{\cC}(\tilde{x},\tilde{y})=s$. We show that the resulting estimate $\tilde{c}:=(\tilde{x},\tilde{y})$ must lie in $c+\im\partial_2^{\cC}$, as desired.
  
  Note that the random cyclic shift in each iteration mentioned above introduces a constant failure probability, so the $\Theta(\log\ell)$ iterations collectively have a $2^{-\Theta(\log\ell)}=1/\poly(\ell)$ success probability. We amplify this sucess probability be repeating the entire algorithm $\poly(\ell)$ times.
\end{proof}

Our decoding algorithm described above shares some similarities with proof of distance in Theorem~\ref{thm:pkdisinf} shown by \cite{panteleev_quantum_2022}. In particular, \cite{panteleev_quantum_2022} also computed prefix sums $\sum_{i\in[k]}X^ic$ of errors $c$ to prove the $\Omega(\ell)$ bound on distance, and relied on expansion of the graphs underlying the Tanner codes $\cA^{(\ell)}$. However, there are significant differences between the two results. Perhaps the most fundamental difference lies in the iterative nature of our decoding algorithm in Theorem~\ref{thm:lpdecinf2}, which is absent from the distance proof of \cite{panteleev_quantum_2022}. Interestingly, the greedy flip-style decoding algorithms for asymptotically good qLDPC codes more closely mirror their distance proofs (see e.g.~\cite{leverrier_decoding_2023,dinur_good_2023}).

\section{Classical Tanner Codes and Noisy-Syndrome Decodability}
\label{sec:classtan}
In this section, we describe classical Tanner codes obtained by imposing the constraints of a good local code around each vertex of an expander graph \cite{sipser_expander_1996}. Furthermore, we follow \cite{panteleev_quantum_2022} in constructing these codes to respect a large cyclic group symmetry. We show that the chain complexes from these codes, along with their associated cochain complexes, are noisy-syndrome decodable in linear time (see Section~\ref{sec:expdecclass}). Applying such chain complexes in a lifted product construction yields the nearly-linear distance codes of \cite{panteleev_quantum_2022}, which we show how to decode in Section~\ref{sec:lpdec}.

Distance proofs for Tanner codes have been previously established (e.g.~\cite{sipser_expander_1996}), while decoding algorithms for such codes have been given both without (e.g.~\cite{sipser_expander_1996,zemor_expander_2001,leverrier_decoding_2023}) and with (e.g.~\cite{spielman_linear-time_1996,gu_single-shot_2023}) syndrome noise. However, for completeness we provide proofs of the noisy-syndrome decodability properties we require in Appendix~\ref{sec:noisysyn}, as we were unable to find the precise statements we need in the prior literature. For instance, \cite{spielman_linear-time_1996} proves noisy-syndrome decodability for the chain complexes associated to classical Tanner codes, but not for the associated cochain complexes. Meanwhile, \cite{gu_single-shot_2023} proves an analogue of noisy-syndrome decodability (i.e.~\textit{single-shot decodability}) for the 3-term chain and cochain complexes associated to quantum Tanner codes \cite{leverrier_quantum_2022-1,leverrier_decoding_2023}, but not for the 2-term complexes associated to classical Tanner codes.

\begin{remark}
  We use classical LDPC codes given by imposing local codes around each vertex of a spectral expander graph \cite{sipser_expander_1996}, because as shown by \cite{panteleev_quantum_2022}, such codes can be made to respect a large cyclic group symmetry using known symmetric spectral expander constructions \cite{agarwal_expansion_2019,jeronimo_explicit_2022}. However, as mentioned above, the chain and cochain complexes of such codes have slightly different structures. In contrast, classical LDPC codes can alternatively be constructed using a different type of expanders called \textit{twosided lossless expanders}. The chain and cochain complexes for such codes have the same structure and hence can be decoded with the same algorithm (see e.g.~\cite{sipser_expander_1996}). However, it remains an open question to construct twosided lossless expanders respecting a large group symmetry (see for instance \cite{capalbo_randomness_2002,lin_good_2022,cohen_hdx_2023,golowich_new_2024}).
\end{remark}

The result below describes the specific 2-term chain complexes we need, which we will construct using classical Tanner codes.

\begin{proposition}
  \label{prop:classtan}
  There exist positive constants $w,\alpha,\beta,e_0,e_1,\gamma$ for which the following holds: there is an explicit infinite family of 2-term chain complexes $\cA^{(\ell)}=(A_1^{(\ell)}\xrightarrow{\partial^{(\ell)}}A_0^{(\ell)})$ indexed by $\ell\in\bN$, such that each $\cA^{(\ell)}$ has the following properties:
  \begin{enumerate}
  \item $\cA^{(\ell)}$ is a 2-term chain complex over $R_\ell$ of locality $w$.
  \item Letting $n_i=\dim_{R_\ell}A^{(\ell)}_i$, then $n_0,n_1,n_1-n_0=\Theta(\log\ell)$ (so that $\dim_{\bF_2}A^{(\ell)}_i=\Theta(\ell\log\ell)$).
  \item Both the chain complex $\cA^{(\ell)}_*$ and its cochain complex ${\cA^{(\ell)}}^*$ are $(\alpha,\beta)$-expanding, and are $(e_0,e_1,\gamma)$-noisy-syndrome decodable in time $O(\ell\log\ell)$ for $e_0,e_1=\Theta(\ell\log\ell)$.
  \end{enumerate}
\end{proposition}

\cite{panteleev_quantum_2022} showed that the 2-term chain complexes such as those in Proposition~\ref{prop:classtan} can be used to obtain lifted product quantum codes of nearly linear distance. However, the 2-term chain complexes used in \cite{panteleev_quantum_2022} were based on expander graphs from random abelian lifts \cite{agarwal_expansion_2019}. These random expanders were later derandomized in \cite{jeronimo_explicit_2022}, which we use to obtain the explicit construction in Proposition~\ref{prop:classtan}.
Theorem~\ref{thm:pkdisinf} above 
follows directly by applying the results of \cite{panteleev_quantum_2022} with the graphs constructed in \cite{jeronimo_explicit_2022}. Note that all of our results also go through using the random abelian lifts of \cite{agarwal_expansion_2019}, at the cost of explicitness.


One of our main contributions (Theorem~\ref{thm:lpdecinf}, proven in Section~\ref{sec:lpdec}) is to efficiently decode the codes in Theorem~\ref{thm:pkdis} up to adversarial corruptions of linear weight with respect to the code's distance.

We prove Proposition~\ref{prop:classtan} in Appendix~\ref{sec:noisysyn}. In particular, Appendix~\ref{sec:noisysyn} contains an essentially self-contained proof of the linear-time noisy-syndrome decodability of the chain and cochain complexes associated to classical Tanner codes. This noisy-syndrome decodability is the main property we need that \cite{panteleev_quantum_2022} does not consider.

\section{Hypergraph Product of Tanner Code and Repetition Code}
\label{sec:hgpdec}
In this section, we present an efficient decoding algorithm for the hypergraph product (see Section~\ref{sec:hgpdef}) of an expanding classical LDPC code with a repetition code. For instance, the hypergraph product of a $[N,\Theta(N),\Theta(N)]$ Tanner code with a $[\Theta(N),1,\Theta(N)]$ repetition code is a $[[\Theta(N^2),\Theta(N),\Theta(N)]]$ quantum code, for which our decoding algorithm can efficiently recover from $\Theta(N)$ adversarial errors.

\subsection{Result Statement}
Our main result on decoding the hypergraph product of a Tanner code and a repetition code is stated below.

\begin{theorem}
  \label{thm:hgpdec}
  Let $\cA_*=(A_1\xrightarrow{\partial^{\cA}}A_0)$ be a 2-term chain complex of locality $w$ such that $\cA_*$ and $\cA^*$ are both $(e_0,e_1,\gamma)$-noisy-syndrome decodable in time $T$. For some $\ell\in\bN$, let $\cB_*=(B_1\xrightarrow{\partial^{\cB}}B_0)$ be the 2-term chain complex over $R_\ell$ given by $B_0=B_1=R_\ell$ and $\partial^{\cB}=1+X\in R_\ell$. Let
  \begin{equation*}
    \cC_* = \cA\otimes_{\bF_2}\cB
  \end{equation*}
  denote the hypergraph product. Let $d(\cC)$ denote the distance of the quantum code associated to $\cC_*$, and let $N^{\cA}=\dim A_0+\dim A_1$. Then for every $\delta>0$, there exists a randomized decoder for this quantum code that protects against $e=\min\{e_0,\;e_1,\;\ell/2,\;(d(\cC)-1)/(2+(w+2)\gamma)\}$ adversarial errors, runs in time $O((N^{\cA}+T)\ell\log(1/\delta))$, and has failure probability $\leq\delta$. Furthermore, there exists a deterministic decoder protecting against $e$ errors that runs in time $O((N^{\cA}+T)\ell^2)$.
\end{theorem}

While we define the chain complex $\cB$ as a chain complex over $R_\ell$ for notational convenience, in this section we do not assume that $\cA$ has an $R_\ell$-structure, and we take the hypergraph product over $\bF_2$. We do consider decoding a lifted product over $R_\ell$ in Section~\ref{sec:lpdec}.

The chain complex $\cB_*=(R_\ell\xrightarrow{1+X}R_\ell)$ and its cochain complex $\cB^*=(R_\ell\xleftarrow{1+X^{\ell-1}}R_\ell)$ both have associated classical code given by the length-$\ell$ repetition code. The parity checks simply check that every pair of consecutive code bits are equal. We reiterate that $X$ here denotes an indeterminate variable in polynomials, rather than a Pauli operator.

\begin{example}
  As an example application of Theorem~\ref{thm:hgpdec} above, consider an instantiation of $\cC=\cA\otimes_{\bF_2}\cB$ where $\cA$ is the 2-term chain complex associated to asymptotically good classical Tanner codes such as those given in Proposition~\ref{prop:classtan}, and $\cB$ is defined by letting $\ell=\Theta(N^{\cA})$.\footnote{We emphasize that this $\ell=\Theta(N^{\cA})$ is just used to define $\cB$, and is not the same $\ell$ as in the statement of Proposition~\ref{prop:classtan}. In fact, Proposition~\ref{prop:classtan} is stronger than necesssary for our purposes here, as we will just use the $\bF_2$-structure of $\cA$, and ignore the cyclic group structure of $\cA$ described in Proposition~\ref{prop:classtan} until we consider lifted products in Section~\ref{sec:lpdec}.} Note that we could alternatively let $\cA$ be the 2-term chain complex associated to a random lossless expander code (see e.g.~\cite{sipser_expander_1996}).

  Then the K\"{u}nneth formula (Proposition~\ref{prop:kunneth}) implies that $\cC$ has associated quantum code dimension
  \begin{equation*}
    \dim(H_1(\cC)) = \dim(H_1(\cA))\cdot\dim(H_0(\cB)) + \dim(H_0(\cA))\cdot\dim(H_1(\cB)) = \Theta(N^{\cA}),
  \end{equation*}
  as $\dim(H_1(\cA))=\dim\ker(\partial^{\cA})\geq\dim(A_1)-\dim(A_0)=\Theta(N^{\cA})$ by Proposition~\ref{prop:classtan}, $\dim(H_0(\cA))\leq O(N^{\cA})$ by definition, $\dim(H_0(\cB))=\ell-\dim(\im\partial^{\cB})=1$ because $\im\partial^{\cB}$ is by definition the set of all vectors in $\bF_2^\ell$ of even Hamming weight, and $\dim(H_1(\cB))=\dim(H^0(\cB))=\dim(H_0(\cB))=1$ by the symmetry of $\cB$.

  Meanwhile, the classical codes associated to $\cA_*$ and $\cA^*$ have distance $\Theta(N^{\cA})$ by Proposition~\ref{prop:classtan}, and the classical codes associated to $\cB_*$ and $\cB^*$ are both length-$\ell$ repetition codes, which have distance $\ell=\Theta(N^{\cA})$. Thus Proposition~\ref{prop:hgpdistance} implies that the quantum code associated to $\cC$ has distance
  \begin{equation*}
    d(\cC) \geq \Theta(N^{\cA}).
  \end{equation*}
  The above bound is in fact tight up to constant factors; see for instance Theorem~15 of \cite{tillich_quantum_2014}.

  Thus $\cC$ is a $[[N=\Theta(N^\cA)^2,\;K=\Theta(N^{\cA}),\;D=\Theta(N^{\cA})]]$ quantum LDPC code, for which Theorem~\ref{thm:hgpdec} provides a $O(N\log(1/\delta))$-time randomized decoder against $\Theta(D)$ adversarial errors with failure probability $\delta$, as well as a $O(N^{3/2})$-time deterministic decoder against $\Theta(D)$ adversarial errors.
\end{example}

We prove Theorem~\ref{thm:hgpdec} in the following sections. Section~\ref{sec:hgpdecalg} below presents a decoding algorithm for $\cC$ (Algorithm~\ref{alg:dechgp}) that takes as input a parameter $j\in[\ell]$; in Section~\ref{sec:hgpcorrect}, we show that this decoding algorithm succeeds for at least half of all choices of $j\in[\ell]$. Section~\ref{sec:hgpdectime} analyzes the running time of this algorithm. Section~\ref{sec:hgpcombine} then combines these results to prove Theorem~\ref{thm:hgpdec}, simply using repeated runs of Algorithm~\ref{alg:dechgp} for different choices of $j\in[\ell]$.

\subsection{Decoder With Constant Probability of Success}
\label{sec:hgpdecalg}
In this section, we present Algorithm~\ref{alg:dechgp}, which contains the technical core of the decoding algorithm we use to prove Theorem~\ref{thm:hgpdec}. Specifically, Algorithm~\ref{alg:dechgp} can be used to successfully decode the code $\cC_*$ in Theorem~\ref{thm:hgpdec} with probability $\geq 1/2$. In Section~\ref{sec:hgpcombine} below, we show how to amplify this success probability to prove Theorem~\ref{thm:hgpdec}.

By symmetry of the construction
\begin{equation*}
  \cC_* = \cA\otimes\cB = \left(A_1\otimes R_\ell \xrightarrow{\partial^{\cA}\otimes I+I\otimes(1+X)} A_0\otimes R_\ell\oplus A_1\otimes R_\ell \xrightarrow{I\otimes(1+X)+\partial^{\cA}\otimes I} A_0\otimes R_\ell\right).
\end{equation*}
in Theorem~\ref{thm:hgpdec}, it will suffice to provide the decoding algorithm $D_0:C_0\rightarrow C_1$ described in Definition~\ref{def:decoder}, as the decoding algorithm $D^2:C^2\rightarrow C^1$ will be exactly analogous. Note that all tensor products $\otimes$ in the equation above are implicitly taken over $\bF_2$.


We let $D^{\cA}:A_0\rightarrow A_1$ be an $(e_0,e_1,\gamma)$-noisy-syndrome decoder for the classical code associated to $\cA$ that runs in time $\leq T$; such a decoder is assumed to exist in the statement of Theorem~\ref{thm:hgpdec}.



For a given error $c\in C_1=A_0\otimes R_\ell\oplus A_1\otimes R_\ell$ of weight $|c|\leq e$, the decoder $D_0$ receives as input the syndrome $s=\partial_1^{\cC}c\in C_0=A_0\otimes R_\ell$. The decoder's goal is to output an element $\tilde{c}\in C_1$ that lies in the coset $c+B_1(\cC)$. We will define $D_0(s)$ to simply perform several appropriate calls to Algorithm~\ref{alg:dechgp}.

\begin{algorithm}[h]
  \caption{Algorithm to decode hypergraph product of Tanner and repetition code.}
  \label{alg:dechgp}
  \SetKwInOut{Input}{Input}
  \SetKwInOut{Output}{Output}

  \SetKwFunction{FnDecHGP}{DecHGP}
  \SetKwProg{Fn}{Function}{:}{}

  \Input{Syndrome $s\in C_0$ of some error $c=(x,y)\in C_1$, and some $j\in[\ell]$}
  \Output{$\tilde{c}\in C_1$ that will be shown to be a low-weight element of $c+B_1(\cC)$ for at least half of all choices of $j\in[\ell]$, as long as $|c|$ is sufficiently small}

  \Fn{\FnDecHGP{$s,j$}}{
    Decomposing $s=\sum_{i\in[\ell]}s_i\otimes X^i$, then for each $k\in[\ell+1]$, compute
    \begin{equation*}
      \tilde{a}_{j+k}:=D^{\cA}\left(\sum_{i\in[k]}s_{j+i\pmod{\ell}}\right)
    \end{equation*} \\ \label{li:nsdec}
    Define $\tilde{y}\in A_1\otimes R_\ell$ by $\tilde{y}:=\sum_{k\in[\ell]}(\tilde{a}_{j+k+1}-\tilde{a}_{j+k})\otimes X^{j+k}$ \\ \label{li:ytdef}
    Compute the element $\tilde{x}\in A_0\otimes R_\ell$ of minimum weight such that $\partial_1^{\cC}(\tilde{x},\tilde{y})=s$. If no such $\tilde{x}$ exists, return FAIL. \\ \label{li:findx}
    \KwRet{$\tilde{c}:=(\tilde{x},\tilde{y})$} \\ \label{li:hgpret}
  }
\end{algorithm}


Note that in line~\ref{li:ytdef} of Algorithm~\ref{alg:dechgp}, $X^{j+k}=X^{j+k\pmod{\ell}}$ by the definition of $R_\ell=\bF_2[X]/(X^\ell-1)$.

Proposition~\ref{prop:hgpdeccorrect} below implies that if $|c|$ is sufficiently small, then the procedure \FnDecHGP{$s,j$} in Algorithm~\ref{alg:dechgp} outputs a correct low-weight decoding $\tilde{c}\in c+B_1(\cC)$ for at least half of all possible choices of $j\in[\ell]$. Thus if we let $D_0(s)$ simply run \FnDecHGP{$s,j$} for a small number of random $j\in[\ell]$, it will obtain a correct decoding $\tilde{c}\in c+B_1(\cC)$ with high probability. Meanwhile, if we instead let $D_0(s)$ run \FnDecHGP{$s,j$} for all $j\in[\ell]$, it will (with a slightly larger time overhead) deterministically obtain a correct decoding $\tilde{c}$. The precise details for defining and analyzing $D_0$ in both of these cases are provided in Section~\ref{sec:hgpcombine}.


\subsection{Correctness of the Constant-Success-Probability Decoding Algorithm}
\label{sec:hgpcorrect}
In this section, we show correctness for the decoder $D_0$ presented in Section~\ref{sec:hgpdecalg}. Formally, we show the following.

\begin{proposition}
  \label{prop:hgpdeccorrect}
  Define all variables as in Theorem~\ref{thm:hgpdec}. For every error $c\in C_1$ of weight $|c|\leq e$, letting $s=\partial_1^{\cC}c$, then it holds for $\geq \ell/2$ of all $j\in[\ell]$ that the procedure \FnDecHGP{$s,j$} in Algorithm~\ref{alg:dechgp} outputs some $\tilde{c}\in C_1$ satisfying $\partial_1^{\cC}\tilde{c}=s$ and $|\tilde{c}|\leq((w+2)\gamma+1)e$.
\end{proposition}
\begin{proof}
  We first fix some notation. For an $\bF_2$-vector space $V$ and an element $v\in V\otimes_{\bF_2} R_\ell$, we let $v=\sum_{i\in[\ell]}v_i\otimes X^i$ be the natural decomposition.

  Fix any $c=(x,y)\in A_0\otimes R_\ell\oplus A_1\otimes R_\ell$ such that $|c|\leq e$. Recall here that all Hamming weights are computed with respect to the $\bF_2$-structure, by counting the number of nonzero bits of the specified vector in its representation in the chosen $\bF_2$-basis. By assumption $|x|\leq|c|\leq\ell/2$, so it holds for $\geq \ell/2$ of all $j\in[\ell]$ that $x_{j-1\pmod{\ell}}=0$. For the remainder of the proof we assume that $j$ satisfies $x_{j-1\pmod{\ell}}=0$.

  For $k\in[\ell+1]$, let $a_{j+k}=\sum_{i\in[k]}y_{j+i\pmod{\ell}}$ denote the prefix sums of $y$. Intuitively, the $\tilde{a}_{j+k}$ defined in line~\ref{li:nsdec} of Algorithm~\ref{alg:dechgp} provide approximations to these prefix sums. The claim below bounds the error in this approximation.

  \begin{claim}
    \label{claim:hgpaclose}
    It holds that $\sum_{k\in[\ell+1]}|\tilde{a}_{j+k}-a_{j+k}| \leq \gamma|x|$ and $\tilde{a}_j-a_j=\tilde{a}_{j+\ell}-a_{j+\ell}=0$.
  \end{claim}
  \begin{proof}[Proof of Claim~\ref{claim:hgpaclose}]
    By definition each $s_i=x_{i-1\pmod{\ell}}+x_i+\partial^{\cA}y_i$. Therefore for $k\in[\ell+1]$, because $x_{j-1\pmod{\ell}}=0$, it follows that $\sum_{i\in[k]}s_{j+i\pmod{\ell}} = x_{j+k-1\pmod{\ell}}+\partial^{\cA}a_{j+k}$. Then as $|x_{j+k-1\pmod{\ell}}|\leq|x|\leq|c|\leq e_0$ and $|a_{j+k}|\leq|y|\leq|c|\leq e_1$, the $(e_0,e_1,\gamma)$-noisy-syndrome decoder $D^{\cA}$ must yield $\tilde{a}_{j+k}=D^{\cA}(x_{j+k-1\pmod{\ell}}+\partial^{\cA}a_{j+k})$ satisfying
    \begin{equation}
      \label{eq:adifhp}
      |\tilde{a}_{j+k}-a_{j+k}| \leq \gamma|x_{j+k-1\pmod{\ell}}|.
    \end{equation}
    Summing~(\ref{eq:adifhp}) over all $k\in[\ell+1]$ then yields the first claim, while setting $k=0,\ell$ in~(\ref{eq:adifhp}) and applying the assumption $x_{j-1\pmod{\ell}}$ yields the second claim.
  \end{proof}

  Now define $b=\sum_{i\in[\ell]}b_i\otimes X^i\in A_1\otimes R_\ell$ by
  \begin{equation*}
    b_{j+i\pmod{\ell}} := \tilde{a}_{j+i}-a_{j+i}.
  \end{equation*}
  for $i\in[\ell+1]$. Note that $b_j=0$ and $|b|\leq\gamma|x|$ by Claim~\ref{claim:hgpaclose}.

  As the $\tilde{a}_{j+i}$ provide approximations to prefix sums of $y$, then line~\ref{li:ytdef} of Algorithm~\ref{alg:dechgp} takes differences of consecutive $\tilde{a}_{j+i}$ to recover an approximation $\tilde{y}$ to $y$. The claim below bounds the error in this approximation.

  \begin{claim}
    \label{claim:hgpydif}
    $\tilde{y}-y=(I\otimes(1+X)X^{\ell-1})b$.
  \end{claim}
  \begin{proof}
    By definition, it holds for all $i\in[\ell]$ that $y_{j+i\pmod{\ell}}=a_{j+i+1}-a_{j+i}$ and $\tilde{y}_{j+i\pmod{\ell}}=\tilde{a}_{j+i+1}-\tilde{a}_{j+i}$. Thus
    \begin{align*}
      (\tilde{y}-y)_{j+i\pmod{\ell}}
      &= (\tilde{a}_{j+i+1}-a_{j+i+1})-(\tilde{a}_{j+i}-a_{j+i}) \\
      &= b_{j+i+1\pmod{\ell}}-b_{j+i\pmod{\ell}}.
    \end{align*}
    The above equality holding for all $i\in[\ell]$ is equivalent to the claim statement, as $\tilde{y}-y=\sum_k(\tilde{y}-y)_k\otimes X^k$, and $(I\otimes(1+X)X^{\ell-1})b=\sum_k(b_{k+1\pmod{\ell}}+b_k)\otimes X^k$, where $b_{k+1\pmod{\ell}}+b_k=b_{k+1\pmod{\ell}}-b_k$ as we are working over $\bF_2$.
  \end{proof}

  Define $\tilde{x}':=x+(\partial^{\cA}\otimes X^{\ell-1})b$. Then
  \begin{align*}
    \partial_1^{\cC}(\tilde{x}',\tilde{y})
    &= (I\otimes(1+X))\tilde{x}'+(\partial^{\cA}\otimes I)\tilde{y} \\
    &= (I\otimes(1+X))x+(\partial^{\cA}\otimes(1+X)X^{\ell-1})b+(\partial^{\cA}\otimes I)\tilde{y} \\
    &= (I\otimes(1+X))x+(\partial^{\cA}\otimes I)(\tilde{y}+y)+(\partial^{\cA}\otimes I)\tilde{y} \\
    &= (I\otimes(1+X))x+(\partial^{\cA}\otimes I)y \\
    &= s,
  \end{align*}
  where the third equality above holds by Claim~\ref{claim:hgpydif}.

  Thus we have shown that $\tilde{x}'$ is a valid choice of $\tilde{x}$ in line~\ref{li:findx} of Algorithm~\ref{alg:dechgp}. Thus line~\ref{li:findx} is guaranteed to find $\tilde{x}$ of weight at most $|\tilde{x}|\leq|\tilde{x}'|$ such that $\partial_1^{\cC}(\tilde{x},\tilde{y})=s$, so Algorithm~\ref{alg:dechgp} returns $\tilde{c}=(\tilde{x},\tilde{y})$ of weight at most
  \begin{align*}
    |\tilde{c}|
    &= |\tilde{x}|+|\tilde{y}| \\
    &\leq |\tilde{x}'|+|y|+|\tilde{y}-y| \\
    &\leq |x|+|(\partial^{\cA}\otimes X^{\ell-1})b|+|y|+|(I\otimes(1+X)X^{\ell-1})b| \\
    &\leq |x|+w|b|+|y|+2|b| \\
    &\leq |x|+|y|+(w+2)\gamma|x| \\
    &\leq (1+(w+2)\gamma)|c|.
  \end{align*}
  where the second inequality above holds by the definition of $\tilde{x}'$ along with Claim~\ref{claim:hgpydif}, the third inequality holds because $\cA$ has locality $w$ by assumption, the fourth inequality holds by the definition of $b$ along with Claim~\ref{claim:hgpaclose}, and the fifth inequality holds by the definition of $c=(x,y)$.

  Thus as $|c|\leq e$, we have shown that for at least $\ell/2$ of all $j\in[\ell]$, the \FnDecHGP{$s,j$} in Algorithm~\ref{alg:dechgp} returns $\tilde{c}$ of weight at most $|\tilde{c}|\leq((w+2)\gamma+1)e$ with $\partial_1^{\cC}\tilde{c}=s$, as desired.
\end{proof}

\subsection{Running Time of the Constant-Success-Probability Decoding Algorithm}
\label{sec:hgpdectime}
In this section, we analyze the running time Algorithm~\ref{alg:dechgp}. Formally, we show the following.

\begin{proposition}
  \label{prop:hgpdectime}
  Define all variables as in Theorem~\ref{thm:hgpdec}. Then the procedure \FnDecHGP{$s,j$} in Algorithm~\ref{alg:dechgp} always terminates in time $O((N^{\cA}+T)\ell)$.
\end{proposition}

We will just need the following lemma.

\begin{lemma}
  \label{lem:inv1pX}
  Given $\zeta\in R_\ell$, there exists a $O(\ell)$-time algorithm that computes the minimum-weight element $\chi\in R_\ell$ satisfying $(1+X)\chi=\zeta$, or outputs FAIL if no such $\chi$ exists.
\end{lemma}
\begin{proof}
  Decomposing $\chi=\sum_{i\in[\ell]}\chi_iX^i$ and $\zeta=\sum_{i\in[\ell]}\zeta_iX^i$, then the condition $(1+X)\chi=\zeta$ is equivalent to requiring that $\chi_{i-1\pmod{\ell}}+\chi_i=\zeta_i$ for all $i\in[\ell]$. Inductively applying this equality for $i=1,\dots,\ell$ implies that $\chi_i=\chi_0+\sum_{k=1}^i\zeta_k$ for all $i\in[\ell]$.

  Given a choice of $\chi_0\in\{0,1\}$, we may use this formula to compute $\chi=\sum_{i\in[\ell]}\chi_iX^i$ in time $O(\ell)$. Therefore we can simply try both $\chi_0=0$ and $\chi_0=1$, and for both resulting $\chi\in R_\ell$, we can test if $(1+X)\chi=\zeta$. If this test succeeds for either $\chi$, we return whichever has smaller Hamming weight; otherwise, if $(1+X)\chi\neq\zeta$ for both $\chi$, we return FAIL.
\end{proof}

\begin{proof}[Proof of Proposition~\ref{prop:hgpdectime}]
  We analyze each line in Algorithm~\ref{alg:dechgp} separately. Below, recall that $N^{\cA}=\dim A_0+\dim A_1$ (see the statement of Theorem~\ref{thm:hgpdec}).
  \begin{description}
  \item[Line~\ref{li:nsdec}:] This step requires $O(N^{\cA}\ell)$ time to compute the partial sums $\sum_{i\in[k]}s_{j+i\pmod{\ell}}$ for $k\in[\ell+1]$, and then $O(T\ell)$ time to run the decoder $D^{\cA}$ on each such partial sum, for a total running time of $O((N^{\cA}+T)\ell)$.
  \item[Line~\ref{li:ytdef}:] This step by definition requires $O(N^{\cA}\ell)$ time.
  \item[Line~\ref{li:findx}:] This step requires $O(N^{\cA}\ell)$ time. Specifically, as $\partial_1^{\cC}(\tilde{x},\tilde{y})=(I\otimes(1+X))\tilde{x}+(\partial^{\cA}\otimes I)\tilde{y}$, in this step the algorithm must simply compute a minimum-weight $\tilde{x}\in A_0\otimes R_\ell$ satisfying
    \begin{equation}
      \label{eq:findtx}
      (I\otimes(1+X))\tilde{x}=z,
    \end{equation}
    for $z:=s-(\partial^{\cA}\otimes I)\tilde{y}$, or else determine that no such $\tilde{x}$ exists.

    For this purpose, we will simply apply Lemma~\ref{lem:inv1pX}. Specifically, let $N_0=\dim A_0$ and let $\{a_h\}_{h\in[N_0]}$ be the fixed basis of $A_0$ (with respect to which we compute Hamming weights), and let $z=\sum_{h\in[N_0]}a_{h}\otimes z_h$ for $z_h\in R_\ell$ be the decomposition into this basis. Because $I\otimes(1+X)$ preserves each subspace $\{a_h\}\otimes R_\ell$, we may apply Lemma~\ref{lem:inv1pX} independently on each such subspace. Specifically, we apply Lemma~\ref{lem:inv1pX} to $\zeta=z_h$ for each $h\in[N_0]$. Each of these $N_0$ runs take $O(\ell)$ time, for a total time of $O(N_0\ell)\leq O(N^{\cA}\ell)$. If all runs succeed, we let $\tilde{x}_h=\chi$ be the respective outputs, and return $\tilde{x}=\sum_{h\in[N_0]}a_h\otimes\tilde{x}_h$. Otherwise, if any runs of Lemma~\ref{lem:inv1pX} fail, we output FAIL.
    
  \item[Line~\ref{li:hgpret}:] This step by definition requires $O(N^{\cA}\ell)$ time.
  \end{description}
  Summing the above running times, we obtain the desired $O((N^{\cA}+T)\ell)$ time for all 4 steps.
\end{proof}

\subsection{Amplifying the Success Probability}
\label{sec:hgpcombine}
In this section, we combine the results of Section~\ref{sec:hgpcorrect} and Section~\ref{sec:hgpdectime} to prove Theorem~\ref{thm:hgpdec}. The main idea is to simply run algorithm~\ref{alg:dechgp} for multiple different $j$ in order to obtain a correct decoding of the syndrome with high probability (or with probability $1$).

\begin{proof}[Proof of Theorem~\ref{thm:hgpdec}]
  As described in Section~\ref{sec:hgpdecalg}, by the symmetry of $\cC_*$ and $\cC^*$, it suffices to provide a decoding algorithm $D_0:C_0\rightarrow C_1$ as described in Definition~\ref{def:decoder}, as the algorithm $D^2:C^2\rightarrow C^1$ will be exactly analogous.
  
  We first prove the statement in Theorem~\ref{thm:hgpdec} regarding the randomized decoder. Specifically, given $\delta>0$, define a randomized decoder $D_0^{(R)}(s)$ that simply calls \FnDecHGP{$s,j$} in Algorithm~\ref{alg:dechgp} for $K:=\lg(1/\delta)$ iid uniform samples of $j\in[\ell]$, and then returns the minimum weight output $\tilde{c}$ from the $K$ runs that satisfies $\partial_1^{\cC}\tilde{c}=s$. By Proposition~\ref{prop:hgpdectime}, $D_0^{(R)}(s)$ always terminates in time $O((N^{\cA}+T)\ell\log(1/\delta))$. Our goal is to show that if $s=\partial_1^{\cC}c$ is the syndrome of some error $c$ of weight $|c|\leq e$, then with probability $\geq 1-\delta$, the resulting output $\tilde{c}\gets D_0^{(r)}(s)$ satisfies $\tilde{c}\in c+B_1(\cC)$.

  Proposition~\ref{prop:hgpdeccorrect} implies that the probability that all $K$ runs fail to output some $\tilde{c}$ with $\partial_1^{\cC}\tilde{c}=s$ and $|\tilde{c}|\leq((w+2)\gamma+1)e$ is at most $(1/2)^K=\delta$. Thus with probability $\geq\delta$, the output $\tilde{c}$ of $D_0^{(r)}(s)$ will satisfy $\partial_1^{\cC}\tilde{c}=s$ and $|\tilde{c}|\leq((w+2)\gamma+1)e$. Thus $\tilde{c}-c\in Z_1(\cC)$, and furthermore, $|\tilde{c}-c|\leq|\tilde{c}|+|c|\leq((w+2)\gamma+2)e<d(\cC)$ because by definition $e<d(\cC)/(2+(w+2)\gamma)$. Therefore by the definition of quantum code distance, $\tilde{c}-c\notin Z_1(\cC)\setminus B_1(\cC)$, so $\tilde{c}-c\in B_1(\cC)$, as desired.

  We now prove the statement in Theorem~\ref{thm:hgpdec} regarding the deterministic decoder. Specifically, define the deterministic decoder $D_0^{(D)}(s)$ to simply call \FnDecHGP{$s,j$} in Algorithm~\ref{alg:dechgp} for every $j\in[\ell]$, and then return the minimum weight output $\tilde{c}$ from the $\ell$ runs that satisfies $\partial_1^{\cC}\tilde{c}=s$. By Proposition~\ref{prop:hgpdectime}, $D_0^{(D)}(s)$ always terminates in time $O((N^{\cA}+T)\ell^2)$. Our goal is to show that if $s=\partial_1^{\cC}c$ is the syndrome of some error $c$ of weight $|c|\leq e$, then the resulting output $\tilde{c}\gets D_0^{(r)}(s)$ always satisfies $\tilde{c}\in c+B_1(\cC)$.

  Proposition~\ref{prop:hgpdeccorrect} implies that there exists some $j\in[\ell]$ for which $\tilde{c}\gets$\FnDecHGP{$s,j$} satisfies $\partial_1^{\cC}\tilde{c}=s$ and $|\tilde{c}|\leq((w+2)\gamma+1)e$. By the same reasoning given above for the randomized case, any such $\tilde{c}$ must lie in $\tilde{c}\in c+B_1(\cC)$, so $D_0^{(D)}(s)$ is gauranteed to output such a $\tilde{c}\in c+B_1(\cC)$, as desired.
\end{proof}

\section{Lifted Product of Tanner Code and Repetition Code}
\label{sec:lpdec}
In this section, we present an efficient decoding algorithm for the lifted product over $R_\ell$ (see Section~\ref{sec:lpdef}) of an expanding classical LDPC code with a repetition code. Specifically, we show how to decode the $[[N,\Theta(\log N),\Theta(N/\log N)]]$ codes of \cite{panteleev_quantum_2022} (see Theorem~\ref{thm:pkdis}) up $\Theta(N/\log N)$ adversarial errors, that is, up to a linear number of errors in the code's distance.

\subsection{Result Statement}
Our main result on decoding the lifted product over $R_\ell$ of a Tanner code and a repetition code is stated below. For ease of presentation, we restrict attention to the case where $\ell$ is a power of $2$, though we expect our algorithm to generalize to arbitrary $\ell\in\bN$.

\begin{theorem}
  \label{thm:lpdec}
  For some $\ell\in\bN$ that is a power of $2$, let $\cA_*=(A_1\xrightarrow{\partial^{\cA}}A_0)$ be a 2-term chain complex over $R_\ell$ of locality $w$ such that $\cA_*$ and $\cA^*$ are both $(e_0,e_1,\gamma)$-noisy-syndrome decodable in time $T$. Let $\cB_*=(B_1\xrightarrow{\partial^{\cB}}B_0)$ be the 2-term chain complex over $R_\ell$ given by $B_0=B_1=R_\ell$ and $\partial^{\cB}=1+X\in R_\ell$. Let
  \begin{equation*}
    \cC_* = \cA\otimes_{R_\ell}\cB
  \end{equation*}
  denote the lifted product over $R_\ell$. Let $d(\cC)$ denote the distance of the quantum code associated to $\cC_*$, and let $n^{\cA}=\dim_{R_\ell}A_0+\dim_{R_\ell}A_1$. Then for every $0<\epsilon\leq 1/2$ and every $\delta>0$, there exists a randomized decoder for this quantum code that protects against
  \begin{equation}
    \label{eq:elpdec}
    e := \min\left\{\frac{e_0}{2},\; \frac{\epsilon e_1}{48\gamma},\; \frac{\epsilon\ell}{12\gamma},\; \frac{d(\cC)-1}{2(w+2)\gamma/\epsilon+2}\right\}
  \end{equation}
  adversarial errors with failure probability $\leq\delta$ and runs in time $O(\ell^{1+2\epsilon}\cdot(n^{\cA}\ell+T)\cdot\log(1/\delta))$.
\end{theorem}

The following example shows that Theorem~\ref{thm:lpdec} combined with Proposition~\ref{prop:classtan} and Theorem~\ref{thm:pkdis} implies Theorem~\ref{thm:lpdecinf}.

\begin{example}
\label{ex:PK-lifted-product}
  If we instantiate $\cC=\cA\otimes_{R_\ell}\cB$ using $\cA=\cA^{(\ell)}$ defined in Proposition~\ref{prop:classtan} for all $\ell\in\bN$ that are powers of $2$, then $\cC$ yields an explicit $[[N=\Theta(\ell\log\ell),\;K=\Theta(\log\ell),D=\;\Theta(\ell)]]$ quantum LDPC code by Theorem~\ref{thm:pkdis}, as shown by \cite{panteleev_quantum_2022}. Setting $0<\epsilon\leq 1/2$, $\epsilon'>0$ to be (arbitrarily small) positive constants, Theorem~\ref{thm:lpdec} provides a close-to-quadratic, that is, $O(N^{2+O(\epsilon+\epsilon')})$-time randomized decoder against $e=\Theta(\epsilon D)$ adversarial errors with failure probability $\delta=2^{-N^{\epsilon'}}$.
\end{example}

The reader is referred to Section~\ref{sec:lpdecinf} for an overview of the main ideas and motivation behind our decoding algorithm in Theorem~\ref{thm:lpdec}. In particular, Section~\ref{sec:lpdecinf} describes where the techniques used to decode the hypergraph products in Theorem~\ref{thm:hgpdec} break down for the lifted product, and provides an overview of how we address this challenge.

As described in Section~\ref{sec:introdecoverview} and Section~\ref{sec:lpdecinf}, while our decoding algorithm in Theorem~\ref{thm:lpdec} shares some similarities with the distance proof (Theorem~\ref{thm:pkdis}) of \cite{panteleev_quantum_2022}, our algorithm is fundamentally iterative unlike the distance proof of \cite{panteleev_quantum_2022}. Perhaps as a result of these differences, Theorem~\ref{thm:lpdec} does not seem to imply a bound on the distance of the codes $\cC$ of \cite{panteleev_quantum_2022}. Rather, the decoding radius in Theorem~\ref{thm:lpdec} grows with the code distance. Therefore we applied the distance bound of \cite{panteleev_quantum_2022} in Example~\ref{ex:PK-lifted-product} to show Theorem~\ref{thm:lpdecinf}.

We prove Theorem~\ref{thm:lpdec} in the following sections. Section~\ref{sec:lpdecintro} below reduces the task of proving Theorem~\ref{thm:lpdec}, that is successfully decoding with arbitrarily high probability, to the task of ``weak decoding,'' meaning successfully decoding with just inverse polynomial success probability. Section~\ref{sec:lpdecprelim} then introduces basic notions and definitions that we will need to construct such a weak decoder, and provides an outline of our approach. Section~\ref{sec:ampcom} then provides the core technical result we use to construct our weak decoder, and Section~\ref{sec:weakdec} applies this result to construct our weak decoder.

In the remainder of this section, we fix all notation as in Theorem~\ref{thm:lpdec}.

\subsection{Reducing to Weak Decoding}
\label{sec:lpdecintro}
In this section, we reduce the problem of proving Theorem~\ref{thm:lpdec}, which provides a decoder for $\cC$ that succeeds with high probability, to the problem of constructing a ``weak'' decoder that succeeds with inverse polynomial probability. We will then subsequently construct such a weak decoder.

By symmetry of the construction
\begin{equation*}
  \cC_* = \cA\otimes_{R_\ell}\cB = \left(A_1 \xrightarrow{\partial^{\cA}+(1+X)} A_0\oplus A_1 \xrightarrow{(1+X)+\partial^{\cA}} A_0\right).
\end{equation*}
in Theorem~\ref{thm:lpdec}, it will suffice to provide the decoding algorithm $D_0:C_0\rightarrow C_1$ described in Definition~\ref{def:decoder}, as the decoding algorithm $D^2:C^2\rightarrow C^1$ will be exactly analogous.

We let $D^{\cA}:A_0\rightarrow A_1$ be an $(e_0,e_1,\gamma)$-noisy-syndrome decoder for the classical code associated to $\cA$ that runs in time $\leq T$; such a decoder is assumed to exist in the statement of Theorem~\ref{thm:lpdec}.

For a given error $c\in C_1$ of weight $|c|\leq e$, the decoder $D_0$ receives as input the syndrome $s=\partial_1^{\cC}c\in C_0=A_0$. The decoder's goal is to output an element $\tilde{c}\in C_1$ that lies in the coset $c+B_1(\cC)$.

We proceed by first constructing a ``weak'' decoder that successfully decodes with probability $\geq 1/\poly(\ell)$. The desired decoder $D_0$ with high success probability then works by simply repeatedly applying the weak decoder polynomially many times. The theorem below describes our weak decoder.

\begin{theorem}
  \label{thm:lpweakdec}
  For some $\llog\in\bN$, let $\ell=2^\llog$, and let $\cA_*=(A_1\xrightarrow{\partial^{\cA}}A_0)$ be a 2-term chain complex over $R_\ell$ of locality $w$ such that $\cA_*$ and $\cA^*$ are both $(e_0,e_1,\gamma)$-noisy-syndrome decodable in time $T$. Let $\cB_*=(B_1\xrightarrow{\partial^{\cB}}B_0)$ be the 2-term chain complex over $R_\ell$ given by $B_0=B_1=R_\ell$ and $\partial^{\cB}=1+X\in R_\ell$. Let
  \begin{equation*}
    \cC_* = \cA\otimes_{R_\ell}\cB
  \end{equation*}
  denote the lifted product over $R_\ell$. Fix any $0<\epsilon\leq 1$, let $n^{\cA}=\dim_{R_\ell}A_0+\dim_{R_\ell}A_1$, and fix any
  \begin{equation}
    \label{eq:eweak}
    e \leq \min\left\{\frac{e_0}{2},\; \frac{\epsilon e_1}{48\gamma},\; \frac{\epsilon\ell}{12\gamma}\right\}.
  \end{equation}
  Then there exists a $O(\ell(n^{\cA}\ell+T))$-time randomized algorithm that takes as input the syndrome $s=\partial_1^{\cC}c$ of any error $c\in C_1$ of weight $|c|\leq e$, and with probability $\geq(1-\epsilon)^\llog$ outputs some $\tilde{c}\in C_1$ of weight $|\tilde{c}|\leq(2(w+2)\gamma/\epsilon+1)e$ such that $\partial_1^{\cC}\tilde{c}=s$.
\end{theorem}


Note that the bound on $e$ in~(\ref{eq:eweak}) in Theorem~\ref{thm:lpweakdec} excludes the term involving $d(\cC)$ in~(\ref{eq:elpdec}) in Theorem~\ref{thm:lpdec}. The reason is that the proof of Theorem~\ref{thm:lpweakdec} does not rely on the distance of $\cC$. However, as the output $\tilde{c}$ of the decoder in Theorem~\ref{thm:lpweakdec} is only guaranteed to satisfy $\partial_1^{\cC}\tilde{c}=s$ and $|\tilde{c}|\leq O(e)$ (assuming $w,\gamma,\epsilon=\Theta(1)$), we can only ensure that $\tilde{c}$ lies in the desired coset $c+B_1(\cC)$ if $e$ is sufficiently small so that $|\tilde{c}-c|\leq O(e)$ is less than $d(\cC)$.

The algorithm described in Theorem~\ref{thm:lpweakdec} is given by Algorithm~\ref{alg:weakdec} below. We will first prove Theorem~\ref{thm:lpdec} assuming Theorem~\ref{thm:lpweakdec}. The subsequent sections will then be dedicated to proving Theorem~\ref{thm:lpweakdec}.

\begin{proof}[Proof of Theorem~\ref{thm:lpdec}]
  As described above, it suffices to provide a decoding algorithm $D_0:C_0\rightarrow C_1$ as described in Definition~\ref{def:decoder}, as the algorithm $D^2:C^2\rightarrow C^1$ will be exactly analogous. Specifically, we want to show that for a given error $c\in C_1$ of weight $|c|\leq e$, then $D_0(s=\partial_1^{\cC}c)$ outputs some $\tilde{c}\in c+B_1(\cC)$ with probability $\geq 1-\delta$.

  Formally, we define $D_0(s)$ to simply run the algorithm described in Theorem~\ref{thm:lpweakdec}
  \begin{equation*}
    K := \left\lceil\frac{\log\delta}{\log(1-(1-\epsilon)^\llog)}\right\rceil
  \end{equation*}
  times with fresh (independent) randomness in each run, and then return the lowest-weight output $\tilde{c}$ of the $K$ runs that satisfies $\partial_1^{\cC}\tilde{c}=s$. Then $D_0(s)$ can only fail to output some $\tilde{c}$ with $\partial_1^{\cC}\tilde{c}=s$ and $|\tilde{c}|\leq(2(w+2)\gamma/\epsilon+1)e$ if all $K$ runs of the algorithm in Theorem~\ref{thm:lpweakdec} fail to output such a $\tilde{c}$. But Theorem~\ref{thm:lpweakdec} ensures that the probability that any one of the $K$ runs fails to output such a $\tilde{c}$ is $\leq 1-(1-\epsilon)^\llog$, so the probability that all $K$ runs fail is at most $(1-(1-\epsilon)^\llog)^K\leq\delta$.

  Therefore with probability $\geq 1-\epsilon$, we have shown that $D_0(s)$ outputs some $\tilde{c}$ with $\partial_1^{\cC}\tilde{c}=s$ and $|\tilde{c}|\leq(2(w+2)\gamma/\epsilon+1)e$. Therefore $\partial_1^{\cC}(\tilde{c}-c)=0$, that is, $\tilde{c}-c\in Z_1(\cC)$. Futhermore,
  \begin{equation*}
    |\tilde{c}-c| \leq |\tilde{c}|+|c| \leq (2(w+2)\gamma/\epsilon+1)e+e < d(\cC_*),
  \end{equation*}
  where the final inequality above holds as $e<d(\cC)/(2(w+2)\gamma/\epsilon+2)$ by definition. Therefore $\tilde{c}-c\notin Z_1(\cC)\setminus B_1(\cC)$ by the definition of the distance $d(\cC_*)$ of a quantum code, so it follows that $\tilde{c}-c\in B_1(\cC)$, as desired.

  It remains to analyze the running time of $D_0$. By definition, $D_0$ simply executes the $O(\ell(n^{\cA}\ell+T))$-time algorithm in Theorem~\ref{thm:lpweakdec} a total of $K$ times. By definition
  \begin{align*}
    K
    &\leq \frac{\log(1/\delta)}{(1-\epsilon)^\llog}+1 \\
    &= \frac{\log(1/\delta)}{2^{\llog\log(1-\epsilon)/\log(2)}}+1 \\
    &\leq \ell^{2\epsilon}\log(1/\delta)+1,
  \end{align*}
  where the first inequality above holds because $-\log(1-\alpha)\geq\alpha$ for all $0\leq\alpha<1$, and the second inequality holds because $2^\llog=\ell$ and because $-\log(1-\epsilon)/\log(2)\leq 2\epsilon$ for all $0<\epsilon\leq 1/2$. Thus $D_0$ runs in time $O(K\cdot\ell(n^{\cA}\ell+T))\leq O(\ell^{1+2\epsilon}\cdot(n^{\cA}\ell+T)\cdot\log(1/\delta))$, as desired.
\end{proof}

\subsection{Preliminary Notions and Main Ideas for Weak Decoding}
\label{sec:lpdecprelim}
In this section, we introduce the main ideas and notions that we will use to prove Theorem~\ref{thm:lpweakdec}. Recall that a more high-level overview of our techniques here can also be found in Section~\ref{sec:lpdecinf}.

To begin, we need the following notions of proximity for elements of $A_1$.

\begin{definition}
  We define the following notions, where below $\Delta_1,\Delta_2,t\in\bN$ are arbitrary positive integers, and $y,\tilde{y}\in A_1$ are arbitrary elements.
  \begin{itemize}
  \item We say that $\tilde{y}$ is \textbf{$\Delta_1$-compatible} with $y$ if there exists some $u\in A_1$ of weight $|u|\leq\Delta_1$ such that $\tilde{y}-y=(1+X)u$.
  \item We say that $y$ is \textbf{$(\Delta_2,t)$-small} if it holds for all $1\leq i\leq t$ that $|\sum_{j\in[i]}X^jy|\leq\Delta_2$.
  \item We say that $\tilde{y}$ is \textbf{$(\Delta_1,\Delta_2,t)$-approximately compatible} with $y$ if there exists a decomposition $\tilde{y}=y'+v$ for some $y'\in A_1$ that is $\Delta_1$-compatible with $y$ and some $v\in A_1$ that is $(\Delta_2,t)$-small.
  \end{itemize}
\end{definition}

The following lemma shows the key property of the notions defined above.

\begin{lemma}
  \label{lem:appcombound}
  If $\tilde{y}\in A_1$ is $(\Delta_1,\Delta_2,t)$-approximately compatible with $y\in A_1$, then it holds for every $i\in\bN$ that
  \begin{equation*}
    \left|\sum_{j\in[i]}X^j(\tilde{y}-y)\right| \leq 2\Delta_1+\left\lceil\frac{i}{t}\right\rceil\Delta_2.
  \end{equation*}
\end{lemma}
\begin{proof}
  By definition, $(\Delta_1,\Delta_2,t)$-approximate compatibility guarantees that there exist $y',u,v\in A_1$ such that $\tilde{y}=y'+v$ with $y'-y=(1+X)u$, $|u|\leq\Delta_1$, and $v$ is $(\Delta_2,t)$-small. Therefore
  \begin{align*}
    \left|\sum_{j\in[i]}X^j(\tilde{y}-y)\right|
    &\leq \left|\sum_{j\in[i]}X^j(y'-y)\right|+\left|\sum_{j\in[i]}X^jv\right|.
  \end{align*}
  Now the first term on the right hand side above equals $|(\sum_{j\in[i]}X^j)(1+X)u|=|(1+X^i)u|\leq 2|u|\leq 2\Delta_1$. Meanwhile, by partitioning $[i]$ into $\lceil i/t\rceil$ sequences of $\leq t$ consecutive integers, the second term on the right hand side above is bounded above a sum of $\lceil i/t\rceil$ expressions of the form $|\sum_{j=i_0}^{i_1}X^jv|$ for some integers $i_0\leq i_1$ satisfying $|i_0-i_1|<t$. But each such expression satisfies $|\sum_{j=i_0}^{i_1}X^jv|=|\sum_{j\in[i_1-i_0+1]}X^jv|\leq\Delta_2$ as $v$ is $(\Delta_2,t)$-small, so we conclude that
  \begin{align*}
    \left|\sum_{j\in[i]}X^j(\tilde{y}-y)\right|
    &\leq 2\Delta_1+\left\lceil\frac{i}{t}\right\rceil\Delta_2,
  \end{align*}
  as desired.
\end{proof}

We are now ready to present the decoding algorithm $D_0(s)$ for $\cC_*$, where $s=\partial_1^{\cC}c$ is the syndrome of an error $c=(x,y)\in A_0\oplus A_1$ of weight $|c|\leq e=\Theta(\ell)$. At a high level, the decoding algorithm begins with $\tilde{y}_0=0$, which is trivially $(\Delta_1,\Delta_2,1)$-approximately compatible for appropriate $\Delta_1,\Delta_2=\Theta(\ell)$ by the trivial decomposition $\tilde{y}=0=y+y$. We then show how given $\tilde{y}$ that is $(\Delta_1,\Delta_2,t)$-approximately compatible with $y$, we can efficiently construct some $\tilde{y}'$ that, with constant probability, is $(\Delta_1,\Delta_2,2t)$-approximately compatible\footnote{A slight modification of the algorithm would allow $2t$ here to be replaced by $\alpha t$ for any constant $\alpha>1$.} with $y$. Iterating this procedure $\Theta(\log\ell)$ times, we obtain some $\tilde{y}'$ that, with $\geq 1/\poly(\ell)$ probability, is $(\Delta_1,\Delta_2,\Theta(\ell))$-approximately compatible with $y$. We will furthermore show that this $\tilde{y}'$ is in fact $\Delta_1$-compatible with $y$. Finally, we show that we may use such a $\tilde{y}'$ to recover some low-weight $\tilde{c}$ with $\partial_1^{\cC}\tilde{c}=s$. Repeating this procedure $\poly(\ell)$ times to boost the success probability, we obtain a $\poly(\ell)$-time algorithm that recovers some $\tilde{c}\in c+B_1(\cC)$ with probability $\geq 1-2^{-\Omega(\ell)}$, as desired.

\subsection{Amplifying Approximate Compatibility}
\label{sec:ampcom}
To begin, below we present Algorithm~\ref{alg:ampcom}, which amplifies approximate compatiblity. This procedure is the technical core of our decoding algorithm for $\cC$. In Section~\ref{sec:weakdec} below, we will apply this algorithm to prove Theorem~\ref{thm:lpweakdec}.

Algorithm~\ref{alg:ampcom} takes as input the syndrome $s\in A_0$ of some error $c=(x,y)\in A_0\oplus A_1$ for the lifted product code $\cC_*$ defined in Theorem~\ref{thm:lpweakdec}, along with some $\tilde{y}\in A_1$ and some $t\in\bN$ such that $\tilde{y}$ is guaranteed to be $(\Delta_1,\Delta_2,t/2)$-approximately compatible with $y$. The algorithm outputs some $\tilde{y}'\in A_1$, which we will show is $(\Delta_1,\Delta_2,t)$-approximately compatible with $y$ with good probability as long as $|c|$ is sufficiently small. For ease of presentation, we assume that $t|\ell$, so for instance $t$ and $\ell$ may both be powers of $2$.

We use the following notation in Algorithm~\ref{alg:ampcom}. We let $n_1=\dim_{R_\ell} A_1$ and we let $\{\alpha_h\}_{h\in[n_1]}$ be the fixed $R_\ell$-basis for $A_1$ (which by assumption is a free based $R_\ell$-module), so that $\{\alpha_hX^{-i}\}_{h\in[n_1],i\in[\ell]}$ gives the fixed $\bF_2$-basis for $A_1$ with respect to which we compute Hamming weights, where $X^{-i}=(X^i)^*=X^{\ell-i}$. For an element $a\in A_1$, we let $a=\sum_{h\in[n_1],i\in[\ell]}a_{h,i}\alpha_hX^{-i}$ denote the decomposition of $a$ into this basis. For $i\notin[\ell]$, we let $a_{h,i}=a_{h,(i\pmod{\ell})}$.

Note that we use the convention of labeling the $(h,i)$th basis element $\alpha_hX^{-i}$ so that for $k\in\bZ$, we have $X^ka=\sum_{h\in[n_1],i\in[\ell]}a_{h,i+k}\alpha_hX^{-i}$ and thus $(X^ka)_{h,i}=a_{h,i+k}$, a property that will simplify expressions later on.

Also, for a sequence of bits $b_1,\dots,b_t\in\{0,1\}$, we let $\Maj_{i\in[t]}\{b_i\}\in\{0,1\}$ denote the majority bit with ties broken arbitrarily, so that $\Maj_{i\in[t]}\{b_i\}=\bar{b}$ if $|\{i\in[t]:b_i=\bar{b}\}|>t/2$.

\begin{algorithm}[h]
  \caption{Algorithm to amplify approximate compatibility.}
  \label{alg:ampcom}
  \SetKwInOut{Input}{Input}
  \SetKwInOut{Output}{Output}

  \SetKwFunction{FnAmpCom}{AmpCom}
  \SetKwProg{Fn}{Function}{:}{}

  \Input{Syndrome $s\in A_0$ of some error $c=(x,y)\in A_0\oplus A_1$, and $\tilde{y}\in A_1$, $t\in\bN$ such that $\tilde{y}$ is $(\Delta_1,\Delta_2,t/2)$-approximately compatible with $y$}
  \Output{$\tilde{y}'\in A_1$ that will be shown to be $(\Delta_1,\Delta_2,t)$-approximately compatible with $y$ with good (e.g.~constant) probability as long as $|c|$ is sufficiently small}

  \Fn{\FnAmpCom{$s,\tilde{y},t$}}{
    For each $1\leq k\leq t$, compute $\tilde{a}_k:=D^{\cA}(\sum_{i\in[k]}X^i(s-\partial^{\cA}\tilde{y}))$ \\ \label{li:ak}
    Choose a uniformly random $j\in[t]$ \\ \label{li:choosej}
    Define $\tilde{z}\in A_1$ so that for every $h\in[n_1]$, $m\in[\ell/t]$, and $i\in[t]$,
    \begin{equation*}
      \tilde{z}_{h,j+mt+i} := \begin{cases}
        0,&i=0 \\
        (\tilde{a}_{i+1})_{h,j+mt}-(\tilde{a}_{i})_{h,j+mt},&1\leq i\leq t-1
      \end{cases}
    \end{equation*} \\ \label{li:z0}
    Let $\tilde{b}_t:=\tilde{a}_t-\sum_{i\in[t]}X^i\tilde{z}$ \\ \label{li:bt}
    Define $\tilde{r}\in A_1$ so that for every $h\in[n_1]$, $m\in[\ell/t]$, and $i\in[t]$,
    \begin{equation*}
      \tilde{r}_{h,j+mt+i} := \begin{cases}
        \Maj_{k\in[t]}\{(\tilde{b}_t)_{h,j+mt-k}\},&i=0 \\
        0,&1\leq i\leq t-1
      \end{cases}
    \end{equation*} \\ \label{li:tr}
    \KwRet{$\tilde{y}+\tilde{z}+\tilde{r}$}
  }
\end{algorithm}

The proposition below shows that Algorithm~\ref{alg:ampcom} indeed amplifies approximate compatibility with good probability.

\begin{proposition}
  \label{prop:ampcom}
  Define all variables as in Theorem~\ref{thm:lpweakdec}. Let $\Delta_1=2\gamma e/\epsilon$, $\Delta_2=12\gamma e/\epsilon$, and let $t\in\bN$ be an integer such that $2|t$ and $t|\ell$. Fix some error $c=(x,y)\in C_1=A_0\oplus A_1$ of weight $|c|\leq e$. Let $s=\partial_1^{\cC}c$ be the syndrome of $c$, and let $\tilde{y}\in A_1$ be some element that is $(\Delta_1,\Delta_2,t/2)$-approximately compatible with $y$. Then with probability $\geq 1-\epsilon$, Algorithm~\ref{alg:ampcom} outputs some $\tilde{y}'\gets\FnAmpCom{$s,\tilde{y},t$}$ that is $(\Delta_1,\Delta_2,t)$-approximately compatible with $y$. Furthermore, if $t>12\gamma e/\epsilon$, then $\tilde{y}'$ is $\Delta_1$-compatible with $y$.
\end{proposition}
\begin{proof}
  For $1\leq k\leq t$, define
  \begin{equation}
    \label{eq:akdef}
    a_k = \sum_{i\in[k]}X^i(y-\tilde{y}).
  \end{equation}
  Thus the $a_k$ give true prefix sums of $y-\tilde{y}$, while the $\tilde{a}_k$ defined in line~\ref{li:ak} of Algorithm~\ref{alg:ampcom} give approximations to these prefix sums. We begin with the following claim that applies the definition of noisy-syndrome decoding to bound the error in this approximation.
  
  \begin{claim}
    \label{claim:lpakbound}
    It holds for every $1\leq k\leq t$ that $|\tilde{a}_k-a_k|\leq 2\gamma|x|$.
  \end{claim}
  \begin{proof}
    By definition, $\tilde{a}_k:=D^{\cA}(\sum_{i\in[k]}X^i(s-\partial^{\cA}\tilde{y}))$ where
    \begin{align*}
      \sum_{i\in[k]}X^i(s-\partial^{\cA}\tilde{y})
      &= \sum_{i\in[k]}X^i(1+X)x+\sum_{i\in[k]}X^i\partial^{\cA}(y-\tilde{y}) \\
      &= (1+X^k)x+\partial^{\cA}\sum_{i\in[k]}X^i(y-\tilde{y}).
    \end{align*}
    By definition $e\leq e_0/2$ (see Theorem~\ref{thm:lpweakdec}), so $|(1+X^k)x|\leq 2|x|\leq 2e\leq e_0$. Meanwhile, by Lemma~\ref{lem:appcombound} and because $\Delta_1,\Delta_2\leq e_1/4$ as $\Delta_1,\Delta_2\leq 12\gamma e/\epsilon$ (see Proposition~\ref{prop:ampcom}) and $e\leq\epsilon e_1/48\gamma$ (see Theorem~\ref{thm:lpweakdec}), we have $|\sum_{i\in[k]}X^i(y-\tilde{y})|\leq 2\Delta_1+2\Delta_2\leq e_1$. Thus as $D^{\cA}$ is an $(e_0,e_1,\gamma)$-noisy-syndrome decoder for $\cA$, it follows that
    \begin{align*}
      |\tilde{a}_k-a_k|
      &= \biggl|D^{\cA}\Bigl((1+X^k)x+\partial^{\cA}\sum_{i\in[k]}X^i(y-\tilde{y}))\Bigr)-\sum_{i\in[k]}X^i(y-\tilde{y})\biggr| \\
      &\leq \gamma|(1+X^k)x| \\
      &\leq 2\gamma|x|. \quad \qedhere
    \end{align*}
  \end{proof}

  As Claim~\ref{claim:lpakbound} shows that $\tilde{a}_k$ is a good approximation to the true prefix sums $a_k$ of $y-\tilde{y}$, we should expect to be able to approximately recover $y-\tilde{y}$ by looking at differences $\tilde{a}_{k+1}-\tilde{a}_k$ of consecutive such elements, as indeed by definition $a_{k+1}-a_k=X^k(y-\tilde{y})$. The following claim pertaining formalizes this intuition to show that $\tilde{z}$ defined in line~\ref{li:z0} of Algorithm~\ref{alg:ampcom} indeed provides an approximation to $y-\tilde{y}$, at least at all components $(h,j+mt+i)$ for all $i\not\equiv 0\pmod{t}$. Below, recall that we have fixed some $0<\epsilon\leq 1$.
  
  \begin{claim}
    \label{claim:lptzbound}
    It holds with probability $\geq 1-\epsilon$ over the choice of $j$ in line~\ref{li:choosej} in Algorithm~\ref{alg:ampcom} that there exists some $u\in A_1$ of weight $|u|<2\gamma|x|/\epsilon$ and some $r\in A_1$ such that $r_{h,j+mt+i}=0$ for all $h\in[n_1]$, $m\in[\ell/t]$, $1\leq i\leq t-1$, and such that
    \begin{equation*}
      \tilde{z}=y-\tilde{y}+(1+X)u+r.
    \end{equation*}
  \end{claim}
  \begin{proof}
    Define $u\in A_1$ by letting for all $h\in[n_1]$, $m\in[\ell/t]$, $1\leq i\leq t$,
    \begin{equation}
      \label{eq:defu}
      u_{h,j+mt+i} := (\tilde{a}_{i})_{h,j+mt}-(a_{i})_{h,j+mt}.
    \end{equation}
    Recal that $j\in[t]$ is sampled uniformly at random, so that
    \begin{align*}
      \bE_j[|u|]
      &= \frac{1}{t}\sum_{j\in[t],h\in[n_1],m\in[\ell/t],1\leq i\leq t}|(\tilde{a}_{i})_{h,j+mt}-(a_{i})_{h,j+mt}| \\
      &= \frac{1}{t}\sum_{1\leq i\leq t}|\tilde{a}_{i}-a_{i}| \\
      &\leq 2\gamma|x|,
    \end{align*}
    where the inequality above holds by Claim~\ref{claim:lpakbound}. Thus by Markov's inequality, we obtain the desired bound on $u$, namely that
    \begin{equation*}
      \Pr[|u|\geq 2\gamma|x|/\epsilon] \leq \epsilon.
    \end{equation*}

    Now condition on choosing $j$ for which $|u|\leq 2\gamma|x|\epsilon$, and define $r\in A_1$ by
    \begin{equation*}
      r := \tilde{z}-(y-\tilde{y})-(1+X)u.
    \end{equation*}
    It remains to be shown that $r_{h,j+mt+i}=0$ for all $h\in[n_1]$, $m\in[\ell/t]$, $1\leq i\leq t-1$. But indeed for every such $h,m,i$, by definition
    \begin{align*}
      \tilde{z}_{h,j+mt+i} &= (\tilde{a}_{i+1})_{h,j+mt}-(\tilde{a}_i)_{h,j+mt} \\
      (y-\tilde{y})_{h,j+mt+i} &= (a_{i+1})_{h,j+mt}-(a_i)_{h,j+mt} \\
      ((1+X)u)_{h,j+mt+i} &= u_{h,j+mt+i+1}-u_{h,j+mt+i} \\
                             &\hspace{1em}= ((\tilde{a}_{i+1})_{h,j+mt}-(a_{i+1})_{h,j+mt})-((\tilde{a}_{i})_{h,j+mt}-(a_{i})_{h,j+mt}),
    \end{align*}
    where the first equality above follows by the definition of $\tilde{z}$ (line~\ref{li:z0} of Algorithm~\ref{alg:ampcom}), the second equality holds by the definition of $a_k$ in~(\ref{eq:akdef}), and the final equality holds by the definition of $u$ in~(\ref{eq:defu}). Combining the above equations immediately gives that for all $h\in[n_1]$, $m\in[\ell/t]$, $1\leq i\leq t-1$, it holds that $r_{h,j+mt+i}=0$, as desired.
  \end{proof}

  Claim~\ref{claim:lptzbound} shows that $\tilde{z}-r$ provides a good approximation to $y-\tilde{y}$ for some $r$ supported in components at indices $(h,mt)$ for $h\in[n_1]$ and $m\in[\ell/t]$. The following claim shows that $\tilde{r}\in A_1$ defined in line~\ref{li:tr} of Algorithm~\ref{alg:ampcom} provides a good approximation to $r$.

  Below, for $0<\epsilon\leq 1$, we let $E_\epsilon$ denote the event occuring with probability $\geq 1-\epsilon$ referenced in the statement of Claim~\ref{claim:lptzbound}. Formally, $E_\epsilon$ is the event that $j\in[t]$ is selected such that $u\in A_1$ defined by~(\ref{eq:defu}) has weight $|u|<2\gamma|x|/\epsilon$.

  \begin{claim}
    \label{claim:lptrbound}
    If event $E_\epsilon$ occurs, then $|\tilde{r}-r|\leq 12\gamma|x|/\epsilon t$.
  \end{claim}
  \begin{proof}
    Let $f=\tilde{a}_t-a_t$, so that by Claim~\ref{claim:lpakbound}, $|f|\leq 2\gamma|x|$. Recall from line~\ref{li:bt} of Algorithm~\ref{alg:ampcom} that $\tilde{b}_t:=\tilde{a}_t-\sum_{i\in[t]}X^i\tilde{z}$. As $\tilde{a}_t=a_t+f=\sum_{i\in[t]}X^i(y-\tilde{y})+f$, then applying the expression for $\tilde{z}$ in Claim~\ref{claim:lptzbound}, we obtain
    \begin{align*}
      \tilde{b}_t
      &= \sum_{i\in[t]}X^i(y-\tilde{y})+f - \sum_{i\in[t]}X^i(y-\tilde{y}+(1+X)u+r) \\
      &= \left(\sum_{i\in[t]}X^ir\right) + \left(f+(1+X^t)u\right).
    \end{align*}
    Therefore by definition, for every $h\in[n_1]$, $m\in[\ell/t]$, and $k\in[t]$, we have
    \begin{align*}
      (\tilde{b}_t)_{h,j+mt-k}
      &= \sum_{i\in[t]}r_{h,j+mt-k+i} + (f+(1+X^t)u)_{h,j+mt-k} \\
      &= r_{h,j+mt}+(f+(1+X^t)u)_{h,j+mt-k},
    \end{align*}
    where the second equality above holds because $r_{h,j+mt-k+i}=0$ for all $k\neq i\in[t]$ by the definition of $r$ in Claim~\ref{claim:lptzbound}. Thus for a given $h\in[n_1]$ and $m\in[\ell/t]$,
    \begin{equation*}
      \tilde{r}_{h,j+mt} := \Maj_{k\in[t]}\{(\tilde{b}_t)_{h,j+mt-k}\}
    \end{equation*}
    can only fail to equal $r_{h,j+mt}$ if $(f+(1+X^t)u)_{h,j+mt-k}\neq 0$ for at least $t/2$ values of $k\in[t]$. It follows that there are a total of at most $|f+(1+X^t)u|/(t/2)$ distinct choices of $(h,m)\in[n_1]\times[\ell/t]$ for which $\tilde{r}_{h,j+mt}\neq r_{h,j+mt}$. As $\tilde{r}$ and $r$ by definition both vanish at all components $(h,j+mt+i)$ for $i\not\equiv 0\pmod{t}$, it follows that
    \begin{equation*}
      |\tilde{r}-r| \leq \frac{|f+(1+X^t)u|}{t/2}.
    \end{equation*}
    Now conditioned on event $E$ occuring, by definition
    \begin{equation*}
      |f+(1+X^t)u| \leq |f|+2|u| \leq 6\gamma|x|/\epsilon.
    \end{equation*}
    Combining the above two inequalities immediately yields the desired claim.
  \end{proof}

  Combining Claim~\ref{claim:lptzbound} with Claim~\ref{claim:lptrbound}, we conclude that the output $\tilde{y}':=\tilde{y}+\tilde{z}+\tilde{r}$ of Algorithm~\ref{alg:ampcom} is of the form
  \begin{align*}
    \tilde{y}'
    &= \tilde{y} + (y-\tilde{y}+(1+X)u+r) + \tilde{r} \\
    &= y+(1+X)u+(\tilde{r}-r).
  \end{align*}
  Furthermore, conditioned on event $E_\epsilon$, which occurs with probability $\Pr[E_\epsilon]\geq 1-\epsilon$, then $|u|<2\gamma|x|/\epsilon\leq 2\gamma e/\epsilon\leq\Delta_1$ (see Proposition~\ref{prop:ampcom} for the definition of $\Delta_1$) and $|\tilde{r}-r|\leq 12\gamma|x|/\epsilon t\leq 12\gamma e/\epsilon t$. Thus if $t>12\gamma e/\epsilon$, then $|\tilde{r}-r|<1$, which implies that $\tilde{r}=r$, and $\tilde{y}'=y+(1+X)u$ for $|u|\leq\Delta_1$, meaning that $\tilde{y}'$ is $\Delta_1$-compatible with $y$. Meanwhile, even without the assumption $t>12\gamma e/\epsilon$, it also follows that for all $1\leq i\leq t$, we have $|\sum_{k\in[i]}X^k(\tilde{r}-r)|\leq t\cdot|\tilde{r}-r|\leq 12\gamma e/\epsilon\leq\Delta_2$, meaning that $\tilde{r}-r$ is $(\Delta_2,t)$-small (see Proposition~\ref{prop:ampcom} for the definition of $\Delta_2$). Thus $\tilde{y}'$ is $(\Delta_1,\Delta_2,t)$-approximately compatible with $y$. These conclusions complete the proof of Proposition~\ref{prop:ampcom}.
\end{proof}

The lemma below analyzes the running time of Algorithm~\ref{alg:ampcom}.

\begin{lemma}
  \label{lem:ampcomtime}
  Define all variables as in Theorem~\ref{thm:lpweakdec}. Running \FnAmpCom{$s,\tilde{y},t$} in Algorithm~\ref{alg:ampcom} takes time $O(t(n^{\cA}\ell+T))$.
\end{lemma}
\begin{proof}
  By definition line~\ref{li:ak} takes time $O(t\cdot n^{\cA}\ell)$ to compute all $\sum_{i\in[k]}X^i(s-\partial^{\cA}\tilde{y})$ for $1\leq k\leq t$, and then time $O(tT)$ to apply $D^{\cA}$ to these $t$ vectors, for a total time of $O(t(n^{\cA}\ell+T))$. Each subsequent line in the algorithm by definition takes $O(t\cdot n^{\cA}\ell)$ time, as each subsequent line simply performs some computation on a constant number of vectors in $A_1$ (specifically $\tilde{z},\tilde{a},\tilde{b},\tilde{r}$), and each such computation involves accessing each of the $O(n^{\cA}\ell)$ components of these vectors at most $O(t)$ times each. Thus the total running time of Algorithm~\ref{alg:ampcom} is $O(t(n^{\cA}\ell+t))$, as desired.
\end{proof}

\subsection{Weak Decoding via Repeated Approximate-Compatibility Amplification}
\label{sec:weakdec}
In this section, we apply the results of Section~\ref{sec:ampcom} to prove Theorem~\ref{thm:lpweakdec}. Specifically, Algorithm~\ref{alg:weakdec} provides the desired weak decoding algorithm described in Theorem~\ref{thm:lpweakdec}. Recall here that we assume $\ell=2^\llog$ for some integer $\llog$.

\begin{algorithm}[h]
  \caption{Weak decoding algorithm to decode $\cC_*$ described in Theorem~\ref{thm:lpweakdec}. This algorithm is based on repeated applications of Algorithm~\ref{alg:ampcom}.}
  \label{alg:weakdec}
  \SetKwInOut{Input}{Input}
  \SetKwInOut{Output}{Output}

  \SetKwFunction{FnWeakDec}{WeakDec}
  \SetKwProg{Fn}{Function}{:}{}

  \Input{Syndrome $s\in A_0$ of some error $c=(x,y)\in A_0\oplus A_1$}
  \Output{$\tilde{c}=(\tilde{x},\tilde{y})\in A_0\oplus A_1$ that will be a relatively low-weight element of $c+B_1(\cC)$ with probability $\geq 1/\poly(\ell)$, as long as $|c|$ is sufficiently small}

  \Fn{\FnWeakDec{$s$}}{
    Let $\tilde{y}_0:=0$ and $\llog:=\lg\ell$ \\ \label{li:wdinit}
    \For{$\tau\in[\llog]$}{ \label{li:fort}
      Let $\tilde{y}_{\tau+1} \gets$ \FnAmpCom{$s,\tilde{y}_\tau,2^{\tau+1}$} \\ \label{li:tyt}
    }
    Compute the element $\tilde{x}\in A_0$ of minimum weight such that $\partial_1^{\cC}(\tilde{x},\tilde{y}_\llog)=s$. If no such $\tilde{x}$ exists, return FAIL \\ \label{li:wdtx}
    \KwRet{$\tilde{c}:=(\tilde{x},\tilde{y}_\llog)$} \\ \label{li:wdret}
  }
\end{algorithm}

\begin{proof}[Proof of Theorem~\ref{thm:lpweakdec}]
  We will show that Algorithm~\ref{alg:weakdec} satisfies the conditions of the theorem. Fix an error $c=(x,y)\in A_0\oplus A_1$ of weight $|c|\leq e$. Define $\Delta_1$, $\Delta_2$ as in Proposition~\ref{prop:ampcom}. By Definition $\tilde{y}_0=0$ is $(\Delta_1,\Delta_2,1)$-approximately compatible with $y$, as $0=y+y$ where $y$ is trivially $0$-compatible with $y$, and $|y|\leq e\leq\Delta_2$ by the definition of $\Delta_2$.

  Now by Proposition~\ref{prop:ampcom}, for each $\tau\in[\llog]$, conditioned on $\tilde{y}_\tau$ being $(\Delta_1,\Delta_2,2^{\tau})$-approximately compatible with $y$, then $\tilde{y}_{\tau+1}$ defined in line~\ref{li:tyt} of Algorithm~\ref{alg:weakdec} must be $(\Delta_1,\Delta_2,2^{\tau+1})$-approximately compatible with $y$ with probability $\geq 1-\epsilon$. Hence $\tilde{y}_{\llog-1}$ must be $(\Delta_1,\Delta_2,2^{\llog-1})$-approxmately compatible with $y$ with probability $\geq(1-\epsilon)^{\llog-1}$. Therefore because $\ell>12\gamma e/\epsilon$ by the definition of $e$, Proposition~\ref{prop:ampcom} implies that $\tilde{y}_\llog$ will be $\Delta_1$-compatible with $y$ with probability $\geq(1-\epsilon)^\llog$.

  Now we must show that conditioned on this event that $\tilde{y}_\llog$ is $\Delta_1$-compatible with $y$, then Algorithm~\ref{alg:weakdec} returns some $\tilde{c}=(\tilde{x},\tilde{y}_\llog)$ of weight $|\tilde{c}|\leq(2(w+2)\gamma/\epsilon+1)e$ such that $\partial_1^{\cC}(\tilde{x},\tilde{y}_\llog)=s$. For this purpose, by the definition of $\Delta_1$-compatibility, we have $\tilde{y}_\llog=y+(1+X)u$ for some $u\in A_1$ of weight $|u|\leq\Delta_1$. Thus defining $\tilde{x}'=x+\partial^{\cA}u$, then
  \begin{align*}
    \partial_c^{\cC}(\tilde{x}',\tilde{y}_\llog)
    &= (1+X)(x+\partial^{\cA}u)+\partial^{\cA}(y+(1+X)u) \\
    &= (1+X)x+\partial^{\cA}y \\
    &= s.
  \end{align*}
  Thus $\tilde{x}'$ is a valid choice of $\tilde{x}$ in line~\ref{li:wdtx} of Algorithm~\ref{alg:weakdec}, so the actual choice of $\tilde{x}$ in line~\ref{li:wdtx} will indeed satisfy $\partial_1^{\cC}(\tilde{x},\tilde{y}_\llog)=s$ with $|\tilde{x}|\leq|\tilde{x}'|$. It follows that Algorithm~\ref{alg:weakdec} will return $\tilde{c}=(\tilde{x},\tilde{y}_\llog)$ of weight
  \begin{align*}
    |\tilde{c}|
    &= |\tilde{x}|+|\tilde{y}_\llog| \\
    &\leq |\tilde{x}'|+|y+(1+X)u| \\
    &= |x+\partial^{\cA}u|+|y+(1+X)u| \\
    &\leq |x|+|y|+|\partial^{\cA}u|+|(1+X)u| \\
    &\leq e+(w+2)|u| \\
    &= e+(w+2)\Delta_1 \\
    &\leq (1+(w+2)\cdot 2\gamma/\epsilon)e,
  \end{align*}
  where the third inequality above holds by the assumption that $\cA$ has locality $w$, and the fourth inequality follows by the definition of $\Delta_1$ in Proposition~\ref{prop:ampcom}. We have shown the above inequality holds with probability $\geq(1-\epsilon)^\llog$, as desired in the statement of Theorem~\ref{thm:lpweakdec}.

  It only remains to show that Algorithm~\ref{alg:weakdec} runs in time $O(\ell(n^{\cA}\ell+T))$. By Lemma~\ref{lem:ampcomtime}, for $\tau\in[\llog]$, the call to \FnAmpCom{$s,\tilde{y}_\tau,2^{\tau+1}$} in line~\ref{li:tyt} of Algorithm~\ref{alg:weakdec} takes time $O(2^{\tau+1}(n^{\cA}\ell+T))$. Summing over $\tau\in[\llog]$ and recalling that $\llog=\lg\ell$, we see that lines~\ref{li:fort}--\ref{li:tyt} collectively take time $\sum_{\tau\in[\llog]}O(2^{\tau+1}(n^{\cA}\ell+T))=O(\ell(n^{\cA}\ell+T))$.

  Line~\ref{li:wdinit} and line~\ref{li:wdret} in Algorithm~\ref{alg:weakdec} each by definition take time $O(n^{\cA}\ell)$, so it only remains to analyze the running time of line~\ref{li:wdtx}. But line~\ref{li:wdtx} also takes time $O(n^{\cA}\ell)$ by Lemma~\ref{lem:inv1pX}. Specifically, the requirement that $\partial_1^{\cC}(\tilde{x},\tilde{y}_\llog)=s$ is equivalent to requiring that $(1+X)\tilde{x}=z$ for $z:=s-\partial^{\cA}\tilde{y}_\llog$. Letting $\{\beta_h\}_{h\in[n_0]}$ be the fixed $R_\ell$-basis for $A_0=\bigoplus_{h\in[n_0]}\beta_hR_\ell\cong R_\ell^{n_0}$ (with respect to which we compute Hamming weights), then we may decompose $z=\sum_{h\in[n_0]}z_h\beta_h$ for $z_h\in R_\ell$. Finally, we construct $\tilde{x}=\sum_{h\in[n_0]}\tilde{x}_h\beta_h$, where for each $h$, we let $\tilde{x}_h\in R_\ell$ be the minimum weight element satisfying $(1+X)\tilde{x}_h=z_h$, if such a $\tilde{x}_h$ exists. By Lemma~\ref{lem:inv1pX}, we may compute each $\tilde{x}_h$ or else determine it does not exist in time $O(\ell)$. By repeating for all $h\in[n_0]$, we can compute $\tilde{x}$ or determine it does not exist in time $O(n_0\ell)\leq O(n^{\cA}\ell)$.

  Summing over all lines in Algorithm~\ref{alg:weakdec}, we conclude that the algorithm runs in time $O(\ell(n^{\cA}\ell+T))+O(n^{\cA}\ell)=O(\ell(n^{\cA}\ell+T))$, as desired.
\end{proof}

\section{Acknowledgments}
We thank Anirudh Krishna for helpful suggestions that improved the exposition. We thank Isaac Kim for pointing us to relevant references.

\bibliographystyle{alpha}
\bibliography{library}

\appendix

\section{Classical Interpretation of the Hypergraph Product: Decoding Over Space and Time}
\label{sec:spacetime}
In this section, we provide a purely classical view of our decoding algorithm in Section~\ref{sec:hgpdecinf} for the hypergraph product of an expander-based classical LDPC code $\cA$ with a repetition code $\cB$. Specifically, we show how the repetition code can be viewed as a time axis, and then the decoding problem for the hypergraph product is equivalent to the problem of decoding a classical LDPC code in which syndrome measurements are repeated across different time steps, and both code bits and syndrome bits can be corrupted. A similar viewpoint is often used for decoding quantum codes, such as the surface code, in the presence of syndrome errors. However, we describe this interpretation here as it may be helpful to the reader more familiar with classical error correction, and it provides a useful framework for thinking about our decoding algorithms.

As in Theorem~\ref{thm:hgpdecinf}, let $\cA_*=(A_1\xrightarrow{\partial^{\cA}}A_0)$ with $N_i=\dim A_i$ be a length-$N_1$ classical LDPC code that is noisy-syndrome decodable, let $\cB_*=(B_1\xrightarrow{\partial^{\cB}}B_0)$ be a length-$\ell$ repetition code, and let $\cC_*=\cA\otimes_{\bF_2}\cB$ be the hypergraph product. Fix an error $c=(x,y)\in C_1=\bF_2^{N_0\times\ell}\oplus\bF_2^{N_1\times\ell}$ with syndrome $s=\partial_1^{\cC}c=(I\otimes\partial^{\cB})x+(\partial^{\cA}\otimes I)y$. Let $z_i$ denote the $i$th column of a matrix $z$. Assume for simplicity that $x_{-1}=x_{\ell-1}=0$;
this assumption allows us to ``cut open'' the length-$\ell$ cycle of time steps so that time is linear and not circular.

Define matrices $a\in\bF_2^{N_1\times\ell}$, $q\in\bF_2^{N_0\times\ell}$, $r\in\bF_2^{N_0\times\ell}$ by the following prefix sums for $k\in[\ell]$ (note the slight difference in indexing compared to the proof of Theorem~\ref{thm:hgpdecinf} in Section~\ref{sec:hgpdecinf}):
\begin{align*}
  a_k &= \sum_{i=0}^ky_i \\
  q_k &= \sum_{i=0}^k(\partial^{\cB}x)_i = \sum_{i=0}^k(x_{i-1}+x_i) = x_k \\
  r_k &= \sum_{i=0}^ks_i = q_k+\partial^{\cA}a_k.
\end{align*}

Now we imagine maintaining a codeword of $\ker\partial^{\cA}\subseteq\bF_2^{N_1}$ over $\ell$ time steps $k=0,\dots,\ell-1$. We interpret $y_k\in\bF_2^{N_1}$ as the vector of errors that occured on code bits in time step $k$, so that $(y_k)_h=1$ if code bit $h$ was flipped at time $k$. Then $a_k$ gives the cumulative errors that have occured up to time $k$ since time $0$.

Meanwhile, we also assume the syndrome of $\cA$ is measured at each time step, but that the syndrome measurements are noisy. We interpret $q_k=x_k$ as the vector of errors that occured on syndrome bits in time step $k$, so that $(x_k)_h=1$ if syndrome bit $h$ was flipped at time $k$. The vector $r_k$ gives the noisy measured syndrome at time $k$; it equals the sum of the syndrome error $q_k=x_k$ at time $k$ and the true syndrome $\partial^{\cA}a_k$ of the cumulative code error $a_k$.

To decode this classical space-time system, observe that we cannot hope to precisely recover the matrices $y$ and $x$ representing the true errors that occured. Indeed, the noisy measured syndromes $r_k$ are identical in the following two scenarios:
\begin{enumerate}
\item Code bit $h$ is flipped at some time step $k$ and then flipped again at time $k+1$, while no syndrome errors occur.
\item No code bit errors occur, but a syndrome error of $\partial^{\cA}\1_h$ occurs at time step $k$.
\end{enumerate}

Therefore syndrome errors can simulate code bit errors and vice versa, so we can at best hope to recover $(\tilde{x},\tilde{y})$ that equal $(x,y)$ up to such indistinguishable errors. That is, we can at best hope to recover $(\tilde{x},\tilde{y})$ for which there exists some $z\in\bF_2^{N_1\times\ell}$ such that $\tilde{x}=x+(\partial^{\cA}\otimes I)z$ and $\tilde{y}=y+(I\otimes\partial^{\cB})z$. But this is precisely the problem of recovering $\tilde{c}=(\tilde{x},\tilde{y})$ that differs from $c=(x,y)$ by some $\tilde{c}-c\in\im\partial_2^{\cC}$, which is exactly the quantum decoding problem formalized in Definition~\ref{def:decoder}!

In this framework, our decoding algorithm described in the proof of Theorem~\ref{thm:hgpdecinf} is particularly natural. Specifically, the decoding algorithm simply applies the noisy-syndrome decoder $D^{\cA}$ for $\cA$ to each $r_k$ to compute $\tilde{a}_k:=D^{\cA}(r_k)$ (again, mind the change in indexing compared to Section~\ref{sec:hgpdecinf}). It then computes an estimate $\tilde{y}$ of $y\in\bF_2^{N_1\times\ell}$ by letting $\tilde{y}_k:=\tilde{a}_k-\tilde{a}_{k-1}$, and chooses $\tilde{x}\in\bF_2^{N_0\times\ell}$ to be the minimum weight element satisfying $\partial_1^{\cC}(\tilde{x},\tilde{y})=s$. That is, in our classical space-time framework, the decoder in Theorem~\ref{thm:hgpdecinf} simply applies the noisy-syndrome decoder $D^{\cA}$ to the cumulative code bit errors $a_k$ to estimate $y$, and then estimates $x$ as the (classical) syndrome errors of minimum possible weight that are consistent with this estimate of $y$.

\section{Proof of Classical Noisy-Syndrome Decodability}
\label{sec:noisysyn}
In this section, we prove Proposition~\ref{prop:classtan} from Section~\ref{sec:classtan}. Specifically, here we present a self-contained construction and analysis of linear-time noisy-syndrome decoders for both the chain and cochain complexes associated to classical Tanner codes. Our presentation applies similar techniques as found in works such as \cite{sipser_expander_1996,spielman_linear-time_1996,zemor_expander_2001,leverrier_decoding_2023}. However, as described in Section~\ref{sec:classtan}, we were unable to find the precise statements we need in such prior works, so we prove them here.

For some of the proofs in this section, we give explicit constants for clarity of presentation, but we do not attempt to optimize any of these constants.

\subsection{Expander Graph Preliminaries}
This section provides some necessary preliminaries on graph theory and expander graphs.

For $\Delta\in\bN$, consider a $\Delta$-regular graph $G=(V_G,E_G)$. We will restrict attention to undirected such graphs $G$ for which every edge has weight $1$, though there may be multiple distinct edges between the same pair of vertices. For vertices $i,j\in V_G$, we let $W_G(i,j)$ equal the number of edges from vertex $i$ to vertex $j$. We extend this notation to sets $S,T\subseteq V_G$ of vertices by letting $W_G(S,T)=\sum_{i\in S,j\in T}W_G(i,j)$. Then the condition that $G$ is $\Delta$-regular can be expressed as requiring $W_G(i,V_G)=\Delta$ for all $i\in V_G$.

Recall that $G$ has a \textbf{normalized adjacency matrix} $A_G\in\bR^{V\times V}$ whose $(i,j)$-entry is given by $(A_G)_{ij}=W_G(i,j)/\Delta$. All eigenvalues of $A_G$ always lie in $[-1,1]$, and the largest eigenvalue of $A_G$ equals $1$, with eigenvector given by the all-ones vector.

\begin{definition}
  The \textbf{spectral expansion $\lambda(G)$} of a graph $G=(V,E)$ equals the second largest absolute value an eigenvalue of the normalized adjacency matrix $A_G$. 
\end{definition}

The main property of spectral expanders that we will use is the following version of the well-known \textit{expander mixing lemma}:

\begin{lemma}[Expander Mixing Lemma; see for instance Lemma~4.15 of \cite{vadhan_pseudorandomness_2012}]
  \label{lem:expmix}
  Let $G=(V_G,E_G)$ be a $\Delta$-regular graph. For every subset of vertices $S\subseteq V_G$,
  \begin{equation*}
    W_G(S,S) \leq \left(\lambda(G)+\frac{|S|}{|V_G|}\right)\cdot\Delta|S|.
  \end{equation*}
\end{lemma}

We will construct classical Tanner codes based on the explicit spectrally expanding abelian lifts of \cite{jeronimo_explicit_2022}, described in Proposition~\ref{prop:elift} below. However, we first need to define abelian lifts of graphs. In this paper we restrict attention to graph lifts by the cyclic group of order $\ell$, which we refer to as $\ell$-lifts.

\begin{definition}
  Fix some $\ell\in\bN$ and some base graph $G_0=(V_0,E_0)$ with an oriented edge labeling $L:V_0\times V_0\rightarrow\bZ/\ell\bZ$, meaning that for every edge $\{u,v\}\in E_0$, then $L(u,v)=-L(v,u)$. The \textbf{$\ell$-lift of $G_0$ according to $L$} is the graph $G=(V,E)$ with vertex set $V=V_0\times(\bZ/\ell\bZ)$, and edge set
  \begin{equation*}
    E = \{\{(u,i),(v,i+L(u,v))\}:\{u,v\}\in E_0,i\in\bZ/\ell\bZ\}.
  \end{equation*}
\end{definition}


\begin{proposition}[\cite{jeronimo_explicit_2022}]
  \label{prop:elift}
  For every constant $\lambda>0$, there exists a sufficiently large constant $\Delta_0=\Delta_0(\lambda)$ such that for every constant $\Delta\geq\Delta_0$, there is an explicit infinite family of $\Delta$-regular graphs $G^{(\ell)}$ for $\ell\in\bN$, such that each $G^{(\ell)}$ has spectral expansion $\lambda(G^{(\ell)})\leq\lambda$, and is the $\ell$-lift of some $\Delta$-regular base graph $G_0^{(\ell)}=(V_0,E_0)$ with $|V_0|=\Theta(\log\ell)$.
\end{proposition}

\subsection{Construction of the 2-Term Chain Complexes $\cA^{(\ell)}$}
\label{sec:defA}
In this section, we describe the desired construction of the 2-term chain complexes $\cA^{(\ell)}$ for $\ell\in\bN$ in Proposition~\ref{prop:classtan}. Specifically, we construct $\cA^{(\ell)}=(A_1^{(\ell)}\xrightarrow{\partial^{(\ell)}}A_0^{(\ell)})$ so that the classical code associated to $\cA^{(\ell)}$ is a classical Tanner code on the spectral expander graph $G^{(\ell)}$ given by Proposition~\ref{prop:elift}. This construction was essentially given in \cite{panteleev_quantum_2022}, though as mentioned above we use a different underlying graph $G$.  We will subsequently show that this construction satisfies the criteria of Proposition~\ref{prop:classtan}.

Formally, given constants $\lambda>0$ (to be specified later, in the proof of Proposition~\ref{prop:classtan} in Section~\ref{sec:nsdecproof}), $R=1/4$ (any $0<R<1/2$ should suffice), and $\delta=1/100$ (any $\delta>0$ for which random linear codes of rate $1-R$ and relative distance $\delta$ are achievable by the Gibert-Varshamov bound should suffice), we construct $\cA^{(\ell)}=\cA^{(\ell)}(\lambda,R,\delta)$ as follows.

First, define $\Delta_0(\lambda)$ as in Proposition~\ref{prop:elift}, and choose some sufficiently large constant $\Delta\geq\Delta_0(\lambda)$ such that there exists a full-rank matrix $Z\in\bF_2^{\Gamma\times\Delta}$ with $\Gamma:=\lfloor R\Delta\rfloor$ such that both $\ker Z$ and $\im Z^\top$ are linear codes of length $\Delta$ and distance $\geq\delta\Delta$. Such $\Delta$ and $Z$ must exist by the Gilbert-Varshamov bound (see e.g.~Section~4.2 of \cite{guruswami_essential_2022}), and therefore $Z$ can be found in constant time (and hence explicitly) by a brute force search because $\Delta$ is a constant as $\ell\in\bN$ grows.

Now we define $\cA^{(\ell)}$ to be the chain complex associated to the classical Tanner code (see e.g.~\cite{sipser_expander_1996}) on the $\Delta$-regular graph $G^{(\ell)}$ in Proposition~\ref{prop:elift} with inner code $\ker Z$. Formally, for simplicity in notation fix $\ell\in\bN$, let $\cA=\cA^{(\ell)}$, and let $G=(V,E)=G^{(\ell)}$ be the $\ell$-lift of $G_0=(V_0,E_0)=G_0^{(\ell)}$, so that both $G$ and $G_0$ are $\Delta$-regular.

At a high level, $\cA$ will be defined so that its associated classical code $\ker\partial^{\cA}$ has code bits labeled by $E$, parity check bits labeled by $V\times[\Gamma]$, and parity checks defined by requiring the $\Delta$ code bits on the edges incident to each vertex form a codeword of the ``inner code'' $\ker Z\subseteq\bF_2^\Delta$. Note that for each vertex, we have freedom to permute the bits of this inner code when assigning them to the edges incident to that vertex. We require the choice of this permutation to respect the action of $\bZ/\ell\bZ$ on $G$, which ensures the resulting chain complex $\cA$ respects the action of $\bZ/\ell\bZ$, that is, $\cA$ is a chain complex over $R_\ell$.

Formally, we first choose some labeling of the edges incident to each vertex of $G$ that respects that action of $\bZ/\ell\bZ$. For this purpose, we first choose arbitrary labels for the base graph $G_0$. Specifically, for each edge $e_0=\{u_0,v_0\}\in E_0$ of the base graph $G_0$, we choose some values $r_{e_0,u_0},r_{e_0,v_0}\in[\Delta]$ such that for all edges $(e_0)_i=\{u_0,(v_0)_i\},i\in[\Delta]$ incident to a given $u_0\in V$, then the values $r_{(e_0)_i,u_0}$ for $i\in[\Delta]$ are all distinct. We then lift this labeling to the $\ell$-lift graph $\tilde{G}$ as follows: for every edge $e=\{(u_0,i),(v_0,i+L(u_0,v_0))\}\in\tilde{E}$, we let $r_{e,(u_0,i)}:=r_{e_0,u_0}$ for $e_0=\{u_0,v_0\}$.

We embed these labels $r_{e,u}$ in a $\bF_2$-linear map $M:\bF_2^E\rightarrow\bF_2^V\otimes\bF_2^\Delta=\bF_2^{V\times[\Delta]}$ as follows: for every edge $e=\{u,v\}\in E$, we let the basis vector $\1_e\in\bF_2^E$ be mapped to $M\1_e=\1_{(u,r_{e,u})}+\1_{(v,r_{e,v})}$.

We are now ready to define $\cA$. Specifically, we let $A_0=\bF_2^{V\times[\Gamma]}\cong\bF_2^V\otimes\bF_2^\Gamma$, $A_1=\bF_2^E$, and we define $\partial^{\cA}:A_1\rightarrow A_0$ by
\begin{equation*}
  \partial^{\cA} = (I\otimes Z)\cdot M.
\end{equation*}
In words, $M$ on the right hand side above maps code bits (associated to edges in $E$) to the incident vertices, and then $I\otimes Z$ applies the parity-check matrix $Z$ to all of these code bits incident to a given vertex. Thus $\partial^{\cA}$ is indeed the parity-check matrix of the Tanner code on the graph $G$ with inner code $\ker Z$, as described above.

We can interpret the cochain complex $\cA^*$ as having associated code $\ker\delta^{\cA}=\ker{\partial^{(\cA)}}^\top=M^\top\cdot(I\otimes Z^\top)$ with code bits labeled by elements of $V\times[\Gamma]$, and parity-check bits labeled by edges in $E$. The parity-check matrix $\delta^{\cA}$ first applies $Z^\top:\bF_2^\Gamma\rightarrow\bF_2^\Delta$ to the $\Gamma$ code bits associated to each vertex, thereby encoding these $\Gamma$ code bits into the ``dual inner code'' $\im Z^\top\subseteq\bF_2^\Delta$. The resulting $\Delta$ inner code bits at a vertex are then distributed among the $\Delta$ edges incident to that vertex. Therefore each edge in $E$ receives a bit from each of its two vertices; the parity check at that edge simply requires these two received bits to be equal.


\cite{panteleev_quantum_2022} showed the following properties of this construction $\cA=\cA^{(\ell)}$.\footnote{As mentioned previously, \cite{panteleev_quantum_2022} actually instantiated the construction slightly differently, for instance by using the $\ell$-lifts of \cite{agarwal_expansion_2019} instead of \cite{jeronimo_explicit_2022}, but their proofs carry over to our setting.}

\begin{lemma}[\cite{panteleev_quantum_2022}]
  \label{lem:ctexp}
  For $\ell\in\bN$, the chain complex $\cA^{(\ell)}$ defined above is a well-defined chain complex over $R_\ell$, which has locality $w\leq\Delta$. Furthermore, if $\lambda>0$ is a sufficiently small constant, then there exist constants $\alpha,\beta>0$ such that $\cA^{(\ell)}_*$ and ${\cA^{(\ell)}}^*$ are both $(\alpha,\beta)$-expanding.
\end{lemma}

Thus to complete the proof of Proposition~\ref{prop:classtan}, it only remains to show that $\cA^{(\ell)}_*$ and ${\cA^{(\ell)}}^*$ are both noisy-syndrome decodable. We show this property below.

\subsection{Proof of Noisy-Syndrome Decodability}
\label{sec:nsdecproof}
In this section, we fix $\cA=\cA^{(\ell)}$ for $\ell\in\bN$ as defined in Section~\ref{sec:defA}, and we show that $\cA_*$ and $\cA^*$ are both noisy-syndrome decodable with running time growing linearly in $N^{\cA}=\dim_{\bF_2}A_0+\dim_{\bF_2}A_1$. Specifically, Algorithm~\ref{alg:nsdecp} and Algorithm~\ref{alg:nsdecd} provide the desired noisy-syndrome decoders for $\cA_*$ and $\cA^*$ respectively. Below, we prove the desired properties of these decoders in Lemma~\ref{lem:nsdecp} and Lemma~\ref{lem:nsdecd} respectively; we remark that these results do not rely on the $R_\ell$-structure (that is, the $\bZ/\ell\bZ$-symmetry) of the chain complexes. We then apply these results to complete the proof of Proposition~\ref{prop:classtan}.

\begin{algorithm}[h]
  \caption{Noisy-syndrome decoder for 2-term chain complex $\cA^{(\ell)}_*=\cA_*=(A_1\xrightarrow{\partial^{\cA}}A_0)$ defined in Section~\ref{sec:defA}. This algorithm uses the following notation: given $x\in(\bF_2^\Delta)^V$, we define the \textit{mismatch vector} $m(x)\in A_1=\bF_2^E$ so that for $e=\{u,v\}\in E$, then $m(x)_e=(x_u)_{r_{e,u}}+(x_v)_{r_{e,v}}$ denotes the mismatch between the assignments to edge $e$ of the local views $x_u$ and $x_v$, where $r_{e,u},r_{e,v}\in[\Delta]$ are as defined in Section~\ref{sec:defA}.}
  \label{alg:nsdecp}
  \SetKwInOut{Input}{Input}
  \SetKwInOut{Output}{Output}

  \SetKwFunction{FnNSDecP}{NSDec$_*$}
  \SetKwProg{Fn}{Function}{:}{}

  \Input{Noisy syndrome $s=a_0+\partial^{\cA}a_1\in A_0=(\bF_2^\Gamma)^V$ arising from some code error $a_1\in A_1$ and syndrome error $a_0\in A_0$}
  \Output{Estimate $\tilde{a}_1\in A_1$ of the code error}

  \Fn{\FnNSDecP{$s$}}{
    Initialize $x\in(\bF_2^\Delta)^V$ so that for every $v\in V$, $x_v\in\bF_2^\Delta$ has minimum possible weight subject to $Zx_v=s_v$ \\
    \While{$\exists v\in V,\;y\in(\bF_2^\Delta)^V$ satisfying $y_v\in\ker Z$, $y_u=0\;\forall u\neq v$, and $|m(x+y)|<|m(x)|$}{
      Update $x_v\gets x_v+y_v$
    }
    \KwRet{any $\tilde{a}_1\in A_1$ such that for all $e=\{u,v\}\in E$ with $m(x)_e=0$, then $(\tilde{a}_1)_e=(x_u)_{r_{e,u}}=(x_v)_{r_{e,v}}$} \\
  }
\end{algorithm}

\begin{lemma}
  \label{lem:nsdecp}
  If $\lambda<\delta/16$, then there exist constants $\epsilon_0,\epsilon_1,\gamma>0$ such that for all $\ell\in\bN$, letting $\cA=\cA^{(\ell)}$ and $N_i=\dim A_i$, then Algorithm~\ref{alg:nsdecp} provides a $(\epsilon_0N_0,\epsilon_1N_1,\gamma)$-noisy-syndrome decoder for $\cA_*$ with running time $O(N_0+N_1)$.
\end{lemma}
\begin{proof}
  Let
  \begin{equation*}
    \gamma = 2\Delta,
  \end{equation*}
  and fix any constants $\epsilon_0,\epsilon_1>0$ such that for all $\ell\in\bN$,
  \begin{align*}
    \epsilon_0 &\leq \frac{\lambda|V|}{2(\Delta+1)N_0} \\
    \epsilon_1 &\leq \frac{\lambda|V|}{4(\Delta+1)N_1}.
  \end{align*}
  By construction $\lambda,\Delta$ are constants and $|V|=\Theta(N_0)=\Theta(N_1)$, so the right hand sides above are indeed bounded below by some constants $\epsilon_0,\epsilon_1>0$.
  
  Fix $a_0\in A_0$ of weight $|a_0|\leq\epsilon_0 N_0$ and $a_1\in A_1$ of weight $|a_1|\leq\epsilon_1 N_1$, and let $s=a_0+\partial^{\cA}a_1$. We will first show that the output $\tilde{a}_1$ of \FnNSDecP{$s$} in Algorithm~\ref{alg:nsdecp} satisfies $|\tilde{a}_1-a_1|\leq\gamma|a_0|$. We will then analyze the running time.

  Below, we use the definition of a mismatch vector $m(x)$ described in Algorithm~\ref{alg:nsdecp}.

  Define $\bar{x}\in(\bF_2^\Delta)^V$ to contain the local views at every vertex of the edge labeling $a_1\in\bF_2^E$, so that for every edge $e=\{u,v\}\in E$, we let $\bar{x}_{v,r_{e,v}}=(a_1)_e$. Note that by definition $m(\bar{x})=0$. Also for $x\in(\bF_2^\Delta)^V$, define
  \begin{equation*}
    S(x) = \{v\in V:x_v\neq\bar{x}_v\}.
  \end{equation*}
  The following claim shows that if an appropriate $x\in(\bF_2^\Delta)^V$ has sufficient disagreement with $\bar{x}$, then some local view $x_v$ can be modified to reduce the mismatch $m(x)$.

  \begin{claim}
    \label{claim:redmis}
    Consider any $x\in(\bF_2^\Delta)^V$ such that $Zx_v=s_v$ for all $v\in V$, and such that
    \begin{equation}
      \label{eq:redmis}
      2|a_0| \leq |S(x)| \leq \lambda|V|.
    \end{equation}
    Then there exists some $v\in V$ and $y\in(\bF_2^\Delta)^V$ with $y_v\in\ker Z$, $y_u=0$ for all $u\neq v$, and $|m(x+y)|<|m(x)|$.
  \end{claim}
  \begin{proof}
    Let $S=S(x)$ for our fixed choice of $x$, and let $T=\{v\in V:(a_0)_v\neq 0\}$, so that $|T|\leq|a_0|$. Define $S'=\{v\in S\setminus T:W_G(v,S)<\delta\Delta/4\}$. If $|S'|=0$, then as $|S|\geq 2|a_0|\geq 2|T|$ by assumption, it holds for $\geq|S|/2$ elements $v\in S$ that $W_G(v,S)\geq\delta\Delta/4$, which implies that $W_G(S,S)\geq(|S|/2)\cdot\delta\Delta/4=(\delta/8)\Delta|S|$. But because $|S|\leq\lambda|V|$ by assumption, Lemma~\ref{lem:expmix} implies that $W_G(S,S)\leq 2\lambda\Delta|S|$, which gives a contradiction as $\delta/8>2\lambda$ by the assumption that $\lambda<\delta/16$. Thus the assumption that $|S'|=0$ was false, so there exists some $v\in S'$.

    Because $v\in S\setminus T$, we have $x_v\neq\bar{x}_v$ but $Zx_v=Z(\bar{x}_v+(a_0)_v)=Z\bar{x}_v$. Therefore $x_v-\bar{x}_v$ is a nonzero element of $\ker Z$, so $|x_v-\bar{x}_v|\geq\delta\Delta$ because $\ker Z$ is by definition a code of distance $\geq\delta\Delta$. Thus define $y\in(\bF_2^\Delta)^V$ such that $y_v=\bar{x}_v-x_v$ and $y_u=0$ for all $u\neq v$. Then the mismatches $m(x+y)$ and $m(x)$ can only disagree on edges incident to vertex $v$. Furthermore, by definition $(x+y)_v=\bar{x}_v$, and the definition of $S'$ ensures that out of the $\Delta$ neighbors $u$ of $v$, at least $(1-\delta/4)\Delta$ of these neighbors must have $(x+y)_u=x_u=\bar{x}_u$ and thus $m(x+y)_{\{u,v\}}=m(\bar{x})_{\{u,v\}}=0$, where the final equality applies the fact that $m(\bar{x})=0$.

    Thus $m(x+y)$ has weight $<\delta\Delta/4$ on the edges incident to vertex $v$. Meanwhile, $m(x)$ agrees with $m(x+y)$ on all edges except those incident to $v$, and because $m(x)=m(x+y-y)=m(x+y)-m(y)=m(x+y)-y$, it follows that $m(x)$ disagrees with $m(x+y)$ on $\geq|y|=|\bar{x}_v-x_v|\geq\delta\Delta$ components. Thus $m(x)$ has weight $\geq 3\delta\Delta/4$ on the edges incident to vertex $v$, so
    \begin{equation*}
      |m(x)| \geq |m(x+y)|+3\delta\Delta/4-\delta\Delta/4 = |m(x+y)|+\delta\Delta/2 > |m(x+y)|,
    \end{equation*}
    as desired.
  \end{proof}

  We now apply Claim~\ref{claim:redmis} to show that the output of Algorithm~\ref{alg:nsdecp} satisfies $|\tilde{a}_1-a_1|\leq\gamma|a_0|$. By definition, the while loop in Algorithm~\ref{alg:nsdecp} reduces $|m(x)|$ by at least $1$ in each step, and only changes $x_v$ for a single vertex $v$ in each step. The initial value of $x$ satisfies $|S(x)|\leq|a_0|+2|a_1|$, as for every vertex $v$ that is not incident to any vertex in the support of $a_0$ nor to any edge in the support of $a_1$, then $x_v=\bar{x}_v=0$. Therefore the initial value of $x$ satisfies $|m(x)|\leq\Delta|S(x)|\leq\Delta(|a_0|+2|a_1|)$, so the value of $x$ at every step in the algorithm must satisfy $|S(x)|\leq(\Delta+1)(|a_0|+2|a_1|)\leq\lambda|V|$, where this final inequality applies the assumption that $|a_i|\leq\epsilon_iN_i$ along with the definitions of $\epsilon_0,\epsilon_1$ above.

  Thus the upper bound in~(\ref{eq:redmis}) holds for $x$ at every step in Algorithm~\ref{alg:nsdecp}, so by Claim~\ref{claim:redmis}, the while loop in Algorithm~\ref{alg:nsdecp} will only terminate when the lower bound in~(\ref{eq:redmis}) is violated. That is, the final value of $x$ in Algorithm~\ref{alg:nsdecp} satisfies $|S(x)|<2|a_0|$. Therefore there are at most $\Delta\cdot 2|a_0|$ edges $e=\{u,v\}\in E$ for which at least one of the equalities $(x_u)_{r_{e,u}}=(\bar{x}_u)_{r_{e,u}}=(a_1)_e$ or $(x_v)_{r_{e,v}}=(\bar{x}_v)_{r_{e,v}}=(a_1)_e$ fails to hold; if both of these equations hold, then the output $\tilde{a}_1$ of Algorithm~\ref{alg:nsdecp} must satisfy $(\tilde{a}_1)_e=(a_1)_e$. Thus $|\tilde{a}_1-a_1|\leq 2\Delta|a_0|=\gamma|a_0|$, as desired.

  It only remains to analyze the running time of Algorithm~\ref{alg:nsdecp}. Let $N^{\cA}=N_0+N_1$, and observe that as $\Delta=O(1)$ and $N^{\cA}=\Theta(N_0)=\Theta(N_1)=\Theta(|V|)=\Theta(|E|)$, all vectors that the algorithm computes have length $O(N^{\cA})$. The initialization of $x$ in Algorithm~\ref{alg:nsdecp} by definition performs $O(N^{\cA})$ constant-time operations, and thus takes $O(N^{\cA})$ total time. We can implement the while loop in Algorithm~\ref{alg:nsdecp} by maintaining a list of all vertices $v\in V$ for which there exists some $y\in(\bF_2^\Delta)^V$ with $y_v\in\ker Z$, $y_u=0$ for all $u\neq v$, and $|m(x+y)|<|m(x)|$. It takes $O(N^{\cA})$ time to initialize this list. In each iteration of the while loop that updates $x_v$ for a vertex $v$, we only need to update any list entries corresponding to the $w=O(1)$ neighbors $u$ of $v$. Thus each iteration of the while loop takes constant time, and there are at most $|E|=O(N^{\cA})$ iterations, so the loop takes a total of $O(N^{\cA})$. Finally, the computation of $\tilde{a}_1$ from the final value of $x$ by definition takes $O(N^{\cA})$ time. Thus Algorithm~\ref{alg:nsdecp} as a whole takes $O(N^{\cA})$ time, as desired.
\end{proof}

\begin{algorithm}[h]
  \caption{Noisy-syndrome decoder for 2-term cochain complex ${\cA^{(\ell)}}^*=\cA^*=(A^1\xrightarrow{\delta^{\cA}}A^0)$ defined in Section~\ref{sec:defA}.}
  \label{alg:nsdecd}
  \SetKwInOut{Input}{Input}
  \SetKwInOut{Output}{Output}

  \SetKwFunction{FnNSDecD}{NSDec$^*$}
  \SetKwProg{Fn}{Function}{:}{}

  \Input{Noisy syndrome $s=a^1+\delta^{\cA}a^0\in A^1=\bF_2^E$ arising from some code error $a^0\in A^0$ and syndrome error $a^1\in A^1$}
  \Output{Estimate $\tilde{a}^0\in A_0$ of the code error}

  \Fn{\FnNSDecD{$s$}}{
    Initialize $x=0\in A^0=(\bF_2^\Gamma)^V$ \\
    \While{$\exists v\in V,\;y\in(\bF_2^\Gamma)^V$ such that $y_u=0\;\forall u\neq v$, and $|s-\delta^{\cA}(x+y)|<|s-\delta^{\cA}(x)|$}{
      Update $x_v\gets x_v+y_v$
    }
    \KwRet{$\tilde{a}^0:=x$} \\
  }
\end{algorithm}

\begin{lemma}
  \label{lem:nsdecd}
  If $\lambda<\delta/16$, then there exist constants $\epsilon^1,\epsilon^0,\gamma>0$ such that for all $\ell\in\bN$, letting $\cA=\cA^{(\ell)}$ and $N_i=\dim A^i$, then Algorithm~\ref{alg:nsdecd} provides a $(\epsilon^1N_1,\epsilon^0N_0,\gamma)$-noisy-syndrome decoder for $\cA^*$ with running time $O(N_0+N_1)$.
\end{lemma}
\begin{proof}
  Let
  \begin{equation*}
    \gamma = 4\Delta,
  \end{equation*}
  and fix any constants $\epsilon_0,\epsilon_1>0$ such that for all $\ell\in\bN$,
  \begin{align*}
    \epsilon^1 &\leq \frac{\lambda|V|}{2N_1} \\
    \epsilon^0 &\leq \frac{\lambda|V|}{2(\Delta+1)N_0}.
  \end{align*}
  By construction $\lambda,\Delta$ are constants and $|V|=\Theta(N_0)=\Theta(N_1)$, so the right hand sides above are indeed bounded below by some constants $\epsilon^0,\epsilon^1>0$.

  Fix $a^1\in A^1$ of weight $|a^1|\leq\epsilon^1N_1$ and $a^0\in A^0$ of weight $|a^0|\leq\epsilon^0N_0$, and let $s=a^1+\delta^{\cA}a^0$. We will first show that the output $\tilde{a}^0$ of \FnNSDecD{$s$} in Algorithm~\ref{alg:nsdecd} satisfies $|\tilde{a}^0-a^0|\leq\gamma|a^1|$. We will then analyze the running time.

  For $x\in(\bF_2^\Gamma)^V$, let
  \begin{equation*}
    S(x) = \{v\in V:x_v\neq a^0_v\}
  \end{equation*}
  denote the set of vertices at which $x$ disagrees with $a^0$. Also, for a syndrome $s\in A^1=\bF_2^E$ and a vertex $v\in V$, let $s_v\in\bF_2^\Delta$ denote the restriction of $s$ to the edges incident to vertex $v$.

  The following claim shows that if a given $x\in(\bF_2^\Gamma)^V$ has sufficient disagreement with $a^0$, then some local view $x_v$ can be modified to reduce the syndrome weight $|s-\delta^{\cA}(x)|$.

  \begin{claim}
    \label{claim:redsyn}
    Consider any $x\in(\bF_2^\Gamma)^V$ such that
    \begin{equation}
      \label{eq:redsyn}
      4|a^1| \leq |S(x)| \leq \lambda|V|.
    \end{equation}
    Then there exists some $v\in V$ and $y\in(\bF_2^\Gamma)^V$ such that $y_u=0$ for all $u\neq 0$, and $|s-\delta^{\cA}(x+y)|<|s-\delta^{\cA}(x)|$.
  \end{claim}
  \begin{proof}
    Let $S=S(x)$ for our fixed choice of $x$, and let $T\subseteq V$ be the set of all vertices that are incident to some edge in the support of $a^1\in\bF_2^E$, so that $|T|\leq 2|a^1|$. Define $S'=\{v\in S\setminus T:W_G(v,S)<\delta\Delta/4\}$. If $|S'|=0$, then as $|S|\geq 4|a^1|\geq 2|T|$ by assumption, it holds for $\geq|S|/2$ elements $v\in S$ that $W_G(v,S)\geq\delta\Delta/4$, which implies that $W_G(S,S)\geq(|S|/2)\cdot\delta\Delta/4=(\delta/8)\Delta|S|$. But because $|S|\leq\lambda|V|$ by assumption, Lemma~\ref{lem:expmix} implies that $W_G(S,S)\leq 2\lambda\Delta|S|$, which gives a contradiction as $\delta/8>2\lambda$ by the assumption that $\lambda<\delta/16$. Thus the assumption that $|S'|=0$ was false, so there exists some $v\in S'$.

    Because $v\in S$, we have $x_v\neq a^0_v$, so $Z^\top x_v$ and $Z^\top a^0_v$ are distinct codewords of $\im Z^\top$, and thus $|Z^\top x_v-Z^\top a^0_v|\geq\delta\Delta$ because $\im Z^\top$ is by definition a code of distance $\geq\delta\Delta$. Define $y\in(\bF_2^\Gamma)^V$ such that $y_v=a^0_v-x_v$ and $y_u=0$ for all $u\neq v$. Then by definition the noisy syndromes $s-\delta^{\cA}(x+y)$ and $s=\delta^{\cA}(x)$ can only disagree on edges incident to vertex $v$. Furthermore, by definition $(x+y)_v=a^0_v$, and because $v\in S'$ so that $v\notin T$, we have $a^1_v=0$, so
    \begin{align*}
      (s-\delta^{\cA}(x+y))_v
      &= a^1_v+(\delta^{\cA}a^0)_v-(\delta^{\cA}(x+y))_v \\
      &= \delta^{\cA}(a^0-(x+y))_v \\
      &= Z^\top(a^0_v-(x+y)_v)+\sum_{u\in V:\{u,v\}\in E}(Z^\top(a^0_u-(x+y)_u))_{r_{\{u,v\},u}}\cdot\1_{r_{\{u,v\},v}} \\
      &= \sum_{u\in V:\{u,v\}\in E}(Z^\top(a^0_u-x_u))_{r_{\{u,v\},u}}\cdot\1_{r_{\{u,v\},v}}
    \end{align*}
    By the definition of $S'$, fewer than $\delta\Delta/4$ of the terms in the sum on the right hand side above can be nonzero.

    Thus we have shown that $|(s-\delta^{\cA}(x+y))_v|<\delta\Delta/4$, that is, $s-\delta^{\cA}(x+y)$ has weight $<\delta\Delta/4$ on the edges incident to vertex $v$. Meanwhile, $s-\delta^{\cA}(x)$ agrees with $s-\delta^{\cA}(x+y)$ on all edges except those incident to $v$, and because $s-\delta^{\cA}(x)=s-\delta^{\cA}(x+y)+\delta^{\cA}(y)$, it follows that $s-\delta^{\cA}(x)$ disagrees with $s-\delta^{\cA}(x+y)$ on $\geq|\delta^{\cA}(y)|=|Z^\top y_v|\geq\delta\Delta$ components. Thus $s-\delta^{\cA}(x)$ has weight $\geq 3\delta\Delta/4$ on the edges incident to vertex $v$, so
    \begin{align*}
      |s-\delta^{\cA}(x)|
      &\geq |s-\delta^{\cA}(x+y)|+3\delta\Delta/4-\delta\Delta/4 \\
      &= |s-\delta^{\cA}(x+y)|+\delta\Delta/2 \\
      &> |s-\delta^{\cA}(x+y)|,
    \end{align*}
    as desired.
  \end{proof}

  We now apply Claim~\ref{claim:redsyn} to show that the output of Algorithm~\ref{alg:nsdecd} satisfies $|\tilde{a}^0-a^0|\leq\gamma|a^1|$. By definition, the while loop in Algorithm~\ref{alg:nsdecd} reduces $|s-\delta^{\cA}(x)|$ by at least $1$ in each step, and only changes $x_v$ for a single vertex $v$ in each step. As the algorithm initializes $x=0$, the initial value of $s-\delta^{\cA}(x)$ is simply $s$, so the while loop runs for a total of at most $|s|\leq|a^1|+\Delta|a^0|$ steps. Therefore as $|S(x)|$ is initially $|S(0)|\leq|a^0|$, and $|S(x)|$ increases by at most $1$ for a total of at most $|a^1|+\Delta|a^0|$ of the while loop, it always holds that $|S(x)|\leq|a^1|+(\Delta+1)|a^0|\leq\lambda|V|$, where this final inequality holds by the assumption that $|a^i|\leq\epsilon^iN_i$ for $\epsilon^0,\epsilon^1$ defined above.

  Thus the upper bound in~(\ref{eq:redsyn}) holds for $x$ at every step in Algorithm~\ref{alg:nsdecd}, so by Claim~\ref{claim:redsyn}, the while loop in Algorithm~\ref{alg:nsdecd} will only terminate when the lower bound in~(\ref{eq:redsyn}) is violated. That is, the final value of $x$ in Algorithm~\ref{alg:nsdecd} satisfies $|S(x)|<4|a^1|$. Thus Algorithm~\ref{alg:nsdecd} returns $\tilde{a}^0$ satisfying $|S(\tilde{a}^0)|<4|a^1|$, which by definition implies that $|\tilde{a}^0-a^0|\leq\Delta\cdot|S(\tilde{a}^0)|\leq 4\Delta|a^1|=\gamma|a^1|$, as desired.

  It only remains the analyze the running time of Algorithm~\ref{alg:nsdecd}. This analysis is similar to that of Algorithm~\ref{alg:nsdecp} in the proof of Lemma~\ref{lem:nsdecp}. Specifically, again let $N^{\cA}=N_0+N_1$, so that $\Delta=O(1)$ and $N^{\cA}=\Theta(N_0)=\Theta(N_1)=\Theta(|V|)=\Theta(|E|)$, which implies that all vectors that the algorithm computes have length $O(N^{\cA})$. We only need to show that the while loop in Algorithm~\ref{alg:nsdecp} can be implemented in time $O(N^{\cA})$. But similarly as the proof of Lemma~\ref{lem:nsdecp}, we can maintain a list of all vertices $v\in V$ for which there exists some $y\in(\bF_2^\Gamma)^V$ with $y_u=0$ for all $u\neq v$ such that $|s-\delta^{\cA}(x+y)|<|s-\delta^{\cA}(x)|$. Initializing this list takes time $O(N^{\cA})$, and then for each of the $O(N^{\cA})$ steps in the while loop, we just need to spend constant time updating entries of the list associated to neighbors of whichever vertex $v$ is updated in that step of the loop. Thus Algorithm~\ref{alg:nsdecp} has overall running time $O(N^{\cA})$, as desired.
\end{proof}

We have now shown all the necessary results to prove Proposition~\ref{prop:classtan}.

\begin{proof}[Proof of Proposition~\ref{prop:classtan}]
  The desired chain complexes $\cA^{(\ell)}$ are given in Section~\ref{sec:defA}, where we choose $0<\lambda<\delta/16$ to be a sufficiently small constant such that Lemma~\ref{lem:ctexp} implies that both $\cA^{(\ell)}_*$ and ${\cA^{(\ell)}}^*$ are $(\alpha,\beta)$-expanding for some constants $\alpha,\beta>0$. Lemma~\ref{lem:ctexp} also implies that each $\cA=\cA^{(\ell)}$ is a well-defined chain complex over $R_\ell$ with locality $w\leq\Delta=\Theta(1)$.

  The construction in Section~\ref{sec:defA} ensures that $n_0=\Gamma|V_0|$ and $n_1=|E_0|$, where $G_0=(V_0,E_0)$ is the $\Delta$-regular base graph described in Proposition~\ref{prop:elift}. Thus by Proposition~\ref{prop:elift} we have $n_0,n_1=\Theta(|V_0|)=\Theta(|E_0|)=\Theta(\log\ell)$, while by construction $n_1-n_0=\Delta|V_0|/2-\Gamma|V_0|=(\Delta/2-\lfloor\Delta/4\rfloor)|V_0|=\Theta(|V_0|)=\Theta(\log\ell)$.

  Finally, Lemma~\ref{lem:nsdecp} and Lemma~\ref{lem:nsdecd} imply that $\cA_*$ and $\cA^*$ respectively are $(e_0,e_1,\gamma)$-noisy-syndrome decodable for $e_i=\Theta(n_i\ell)=\Theta(\ell\log\ell)$ and $\gamma=O(1)$.

  Thus we have shown that $\cA=\cA^{(\ell)}$ satisfies all desired properties in the proposition.
\end{proof}

\end{document}